\PassOptionsToPackage{
pdfencoding=auto,
pdfnewwindow=true,
pdfusetitle=true,
bookmarks=true,
bookmarksnumbered=true,
bookmarksopen=true,
pdfpagemode=UseThumbs,
bookmarksopenlevel=1,
pdfpagelabels=true,
breaklinks=true,
colorlinks=true,
}{hyperref}
\documentclass[english,10pt,superscriptaddress,aps,prl,nofootinbib,reprint,twocolumn,letterpaper,longbibliography,tightenlines,final]{revtex4-2}
\usepackage[utf8]{inputenx}
\usepackage[T1]{fontenc}
\usepackage{lmodern}
\usepackage{cfr-lm}

\usepackage[showerrors,immediate]{silence}
\usepackage{microtype}
\microtypecontext{spacing=nonfrench}
\microtypesetup{
expansion={true,nocompatibility},
protrusion={true,nocompatibility},
activate={true,nocompatibility},
tracking=true,
kerning=true,
spacing={true}
}

\usepackage[all=normal,floats=tight,mathspacing=tight,wordspacing=tight,paragraphs=normal,tracking=tight,charwidths=tight,mathdisplays=normal,sections=normal,margins=normal]{savetrees}
\everypar=\expandafter{\the\everypar\loosness=-1 }
\usepackage{paralist}
\usepackage{graphicx}
\usepackage{float}
\usepackage{amssymb,amsthm,amsfonts,amstext,amsmath}
\DeclareMathSymbol{\shortminus}{\mathbin}{AMSa}{"39}
\usepackage{mathtools}
\usepackage{placeins}
\usepackage{accents}

\usepackage[dvipsnames]{xcolor}

\definecolor{googleblue}{RGB}{34, 0, 204}
\definecolor{carmine}{RGB}{150, 0, 24}
\definecolor{darkgreen}{RGB}{0, 80, 0}
\definecolor{medblue}{RGB}{0, 0, 100}

\usepackage[]{hyperref}
\hypersetup{linkcolor=carmine,citecolor=darkgreen,urlcolor=googleblue,anchorcolor=OliveGreen}

\newcommand{\tmem}[1]{{\em #1\/}}
\newcommand{\tmop}[1]{\ensuremath{\operatorname{#1}}}

\newtheorem{definition}{Definition}
\newtheorem{lemma}{Lemma}

\newtheorem{theo}{Theorem}

\newtheorem{innercustomgeneric}{\customgenericname}
\providecommand{\customgenericname}{}
\newcommand{\newcustomtheorem}[2]{%
  \newenvironment{#1}[1]
  {%
   \renewcommand\customgenericname{#2}%
   \renewcommand\theinnercustomgeneric{##1}%
   \innercustomgeneric
  }
  {\endinnercustomgeneric}
}
\newcustomtheorem{customprop}{Proposition}

\def\bra#1{\langle#1|} \def\ket#1{|#1\rangle}

\def\proj#1{\ket{#1}\!\bra{#1}}

\def\id{{\mathbb I}}
\def\be{\begin{equation}}
\def\ee{\end{equation}}
\def\tr{\operatorname{tr}}
\def\A{\mathcal{A}}
\def\B{\mathcal{B}}
\def\C{\mathcal{C}}
\def\Y{\mathcal{Y}}
\def\X{\mathcal{X}}

\begin{document}

\title{Genuine Network Multipartite Entanglement}

\author{Miguel Navascu{\'e}s}
\affiliation{Institute for Quantum Optics and Quantum Information (IQOQI) Vienna, Austrian Academy of Sciences, Boltzmanngasse 3, 1090 Vienna, Austria}
\author{Elie Wolfe}
\affiliation{Perimeter Institute for Theoretical Physics, 31 Caroline St. N., Waterloo, Ontario, Canada, N2L 2Y5}
\author{Denis Rosset}
\affiliation{Perimeter Institute for Theoretical Physics, 31 Caroline St. N., Waterloo, Ontario, Canada, N2L 2Y5}
\author{Alejandro Pozas-Kerstjens}
\affiliation{Departamento de An\'alisis Matem\'atico, Universidad Complutense de Madrid, 28040 Madrid, Spain}
\affiliation{ICFO-Institut de Ci\`encies Fot\`oniques, The Barcelona Institute of Science and Technology, 08860 Castelldefels (Barcelona), Spain}

\begin{abstract}
  The standard definition of genuine multipartite entanglement stems from the need to assess the quantum control over an ever-growing number of quantum systems.
  We argue that this notion is easy to hack: in fact, a source capable of distributing bipartite entanglement can, by itself, generate genuine $k$-partite entangled states for any $k$.
  We propose an alternative definition for genuine multipartite entanglement, whereby a quantum state is \emph{genuinely network $k$-entangled} if it cannot be produced by applying local trace-preserving maps over several (\hspace{-1pt}$k$-\hspace{-1pt}1\hspace{-1pt})-partite states distributed among the parties, even with the aid of global shared randomness.
   We provide analytic and numerical witnesses of genuine network entanglement, and we reinterpret many past quantum experiments as demonstrations of this feature.
\end{abstract}

{\maketitle}

The existence of multipartite quantum states that cannot be prepared locally is at the heart of many communication protocols in quantum information science, such as quantum teleportation~{\cite{teleportation}}, dense coding~{\cite{denseCoding}}, entanglement-based quantum key distribution~{\cite{qkd}} and the violation of Bell inequalities~{\cite{bell,BrunnerRMP}}.
Most importantly, for the last two decades, the ability to entangle an ever-growing number of photons or atoms has been regarded as a benchmark for the experimental quantum control of optical systems~{\cite{sixPhotons,eightPhotons,tenPhotons,hundredsAtoms}}.

Since any multipartite quantum state where two parts share a singlet can be
regarded as ``entangled,'' another, more demanding notion of entanglement was required to assess the progress of quantum technologies.
The accepted answer was genuine multipartite entanglement~{\cite{wrongDef,wrongDef2,wrongDef3}}.
Genuine multipartite entanglement has since become a standard for quantum many-body experiments~{\cite{sixPhotons,eightPhotons,tenPhotons,hundredsAtoms,thousandsAtoms}}.
But, is it a universal measure?

In this paper, we argue the opposite and present an alternative and stronger definition, {\emph{genuine network multipartite entanglement}}, which we formulate in terms of quantum networks~{\cite{quantNetwork}}.
First, we define and compare the two approaches.
Next, we present general criteria to detect genuine network entanglement and discuss the tightness of the bounds so obtained.
Finally, we single out past experiments in quantum optics that can be reinterpreted as stronger demonstrations of genuine network entanglement.

\paragraph{Multipartite entanglement.}
A $n$-partite quantum state can be identified with a bounded Hermitian positive semidefinite operator $\rho$ acting on a composite Hilbert space $\mathcal{H}_1 \otimes \cdots \otimes \mathcal{H}_n$ such that $\tmop{tr} (\rho) = 1$.
Each factor $\mathcal{H}_i$ with $i = 1, \ldots, n$ represents the local Hilbert space of the $i^{\text{th}}$ party.
For a subset $S \subseteq \{ \mathcal{H}_i \}_i$, we denote by $\rho_{(S)} = \tmop{tr}_{\overline{S}} (\rho)$ the density matrix of the reduced state on the subsystems $S$, where $\overline{S}$ is the complement of $S$.
We say that an $n$-partite state is fully separable if it can be written as a convex mixture of product states as follows:
\begin{equation}
  \label{Eq:FullySep} \rho = \sum_j w_j \rho^j_1 \otimes \cdots \otimes \rho^j_n, \qquad \sum_j w_j = 1,
\end{equation}
where the $\{ \rho^j_i \}$ are normalized density matrices and the weights $w_j$ are nonnegative.
If $\rho$ does not admit a decomposition of the
form~{\eqref{Eq:FullySep}}, we say that it is entangled.
The problem with the definition of full separability is that any technology capable of entangling, say, the first two particles could claim the generation of ``entangled states'' composed of arbitrarily many particles.
Indeed, the reader can check that any state $\hat{\rho}$ of the form
\begin{equation}
  \label{Eq:MultiExample} \hat{\rho} = \proj{\phi^+}
  \otimes \rho_{(\mathcal{H}_3, \ldots, \mathcal{H}_n)}, \quad | \phi^+
  \rangle = \frac{| 00 \rangle + | 11 \rangle}{\sqrt{2}}
\end{equation}
does not admit a decomposition of the form of Eq.~{\eqref{Eq:FullySep}}.

In order to address this issue, an extended definition of multipartite separability was proposed~{\cite{wrongDef,wrongDef2,wrongDef3}}.
Intuitively, a state is $k$-partite entangled if, in order to produce it, one must create $k$-partite entangled states and distribute them among the $n$ parties in such a way that no party receives more than one subsystem.
More formally, we say that an $n$-partite state is separable with respect to a partition $S_1 |\dots| S_s$ of $\{\mathcal{H}_1, \dots, \mathcal{H}_n \}$ if it can be expressed as
\begin{equation}
  \rho = \sum_j w_j \rho_{(S_1)}^j \otimes \cdots \otimes \rho_{(S_s)}^j.
\end{equation}
An $n$-partite state is {\tmem{genuinely $k$-partite entangled}} (or has entanglement depth $k$) if it cannot be expressed as a convex combination of quantum states, each of which is separable with respect to at least one
partition $S_1 |S_2 |\dots$ of $\{1, \dots, n\}$ with $|S_{\ell} | \leq k{\shortminus}1$, for all $\ell$.
Using this definition, the state $\hat{\rho}$ in Eq.~{\eqref{Eq:MultiExample}} is certainly genuinely $2$-entangled.
However, $\hat{\rho}$ is \emph{not} genuinely $3$-entangled so long as its marginal $\hat{\rho}_{(\mathcal{H}_3, \ldots, \mathcal{H}_n)}$ is fully separable.

This notion of multipartite entanglement is easy to cheat, as we show next.
For simplicity, we consider a tripartite scenario ($n\,{=}\,3$) and rename the Hilbert spaces $\mathcal{A}$, $\mathcal{B}$ and $\mathcal{C}$; we split $\mathcal{A}$ into three local subsystems $\mathcal{A}'$, $\mathcal{A}''$ and $\mathcal{A}'''$, the same for $\mathcal{B}$ and $\mathcal{C}$.
Now, let \mbox{$\rho_{\mathcal{A}' \mathcal{B}' \mathcal{C}'} = | \phi^+ \rangle \langle \phi^+ |_{\mathcal{A}' \mathcal{B}'} \otimes \proj{0}_{\mathcal{C}'}$}, and similarly $\rho_{\mathcal{A}'' \mathcal{B}'' \mathcal{C}''} = | \phi^+ \rangle \langle \phi^+ |_{\mathcal{B}'' \mathcal{C}''} \otimes \proj{0}_{\mathcal{A}''}$ while \mbox{$\rho_{\mathcal{A}''' \mathcal{B}''' \mathcal{C}'''} = | \phi^+ \rangle \langle \phi^+ |_{\mathcal{C}''' \mathcal{A}'''} \otimes \proj{0}_{\mathcal{B}'''}$}.
Following the same discussion as for $\hat{\rho}$, each of these three states individually is genuinely $2$-entangled but not genuinely $3$-entangled.
However, if we consider those three states \emph{collectively} (i.e., distributed \emph{at the same time}), then the resulting state $\rho_{\mathcal{ABC}} = \rho_{\mathcal{A}' \mathcal{B}' \mathcal{C}'} \otimes \rho_{\mathcal{A}'' \mathcal{B}'' \mathcal{C}''} \otimes \rho_{\mathcal{A}''' \mathcal{B}''' \mathcal{C}'''}$ is genuinely $3$-entangled when considering the partition $\mathcal{A|B|C}$.
Accordingly, the established definition of genuine $k$-partite entanglement is unstable under parallel composition (i.e., under simultaneous distribution of states).

In fact, enough copies of the state $\rho_{\mathcal{A}\mathcal{B}\mathcal{C}}$ enable the distribution of \emph{any} tripartite state using the standard quantum teleportation protocol~{\cite{teleportation}}. Any definition of genuine tripartite entanglement that regarded states like $\rho_{\mathcal{A} \mathcal{B} \mathcal{C}}$ as \emph{not} genuinely tripartite entangled and, at the same time, were stable under composition and local operations and classical communication (LOCC), would thus be necessarily void. Namely, it would not apply to any physical tripartite state.

In this paper, we introduce the concept of {\tmem{genuine network $k$-entanglement}}, an alternative operational definition of multipartite entanglement that is stable under composition and where $\rho_{\mathcal{A} \mathcal{B} \mathcal{C}}$ is not genuinely tripartite entangled. The drawback, as it will be evident from the definition, is that non-genuine network entanglement is not closed under LOCC, but under the subset of LOCC transformations known as Local Operations and Shared Randomness (LOSR)~\cite{Buscemi2012,schmid2019typeindependent}. This set of operations has been argued to be more relevant than LOCC for the study of Bell nonlocality~\cite{wolfe2020quantifying,schmid2020standard}. Note that LOSR is a natural set of operations when the parties being distributed the states are separated in space and do hot hold a quantum memory.

\paragraph{Genuine network entanglement.}
We explain our definition using an adversarial approach. Eve is a vendor selling a source of tripartite quantum states to three honest scientists Alice, Bob and Charlie. Eve pretends that her device produces a valuable entangled tripartite state $\rho_{\A\B\C}$. Unbeknown to the scientists, the source sold to them is actually composed of cheaper components: quantum sources that produce the bipartite entangled states $\sigma_{\A'\B''},\sigma_{\C'\A''},\sigma_{\B'\C''}$, see Figure~\ref{Fig:Network}. Alice receives the $\A',\A''$ subsystems of the states $\sigma_{\A'\B''},\sigma_{\C' \A''}$. Those can in principle interact within Alice's experimental setup, giving rise to a new quantum system $\A$: that is what Alice eventually probes. Similarly, Bob (resp. Charlie) will have access to system $\B$ (resp. $\C$), whose state is the result of a deterministic interaction between systems $\B',\B''$ (resp. $\C',\C''$). In addition, we provide Eve with unlimited shared randomness $\Lambda$ to jointly influence the local operations acting on systems $\A'\A''$, $\B'\B''$ and $\C'\C''$. It is worth remarking that we do not make any assumption on the dimensionality of the ``hidden'' states $\sigma_{\A'\B''},\sigma_{\C'\A''},\sigma_{\B'\C''}$: even if the systems $\A,\B,\C$ accessible to Alice, Bob and Charlie are a qubit each, the Hilbert space dimension of the hidden systems might well be infinite.

\begin{figure}[ht]
  \includegraphics[width=5 cm]{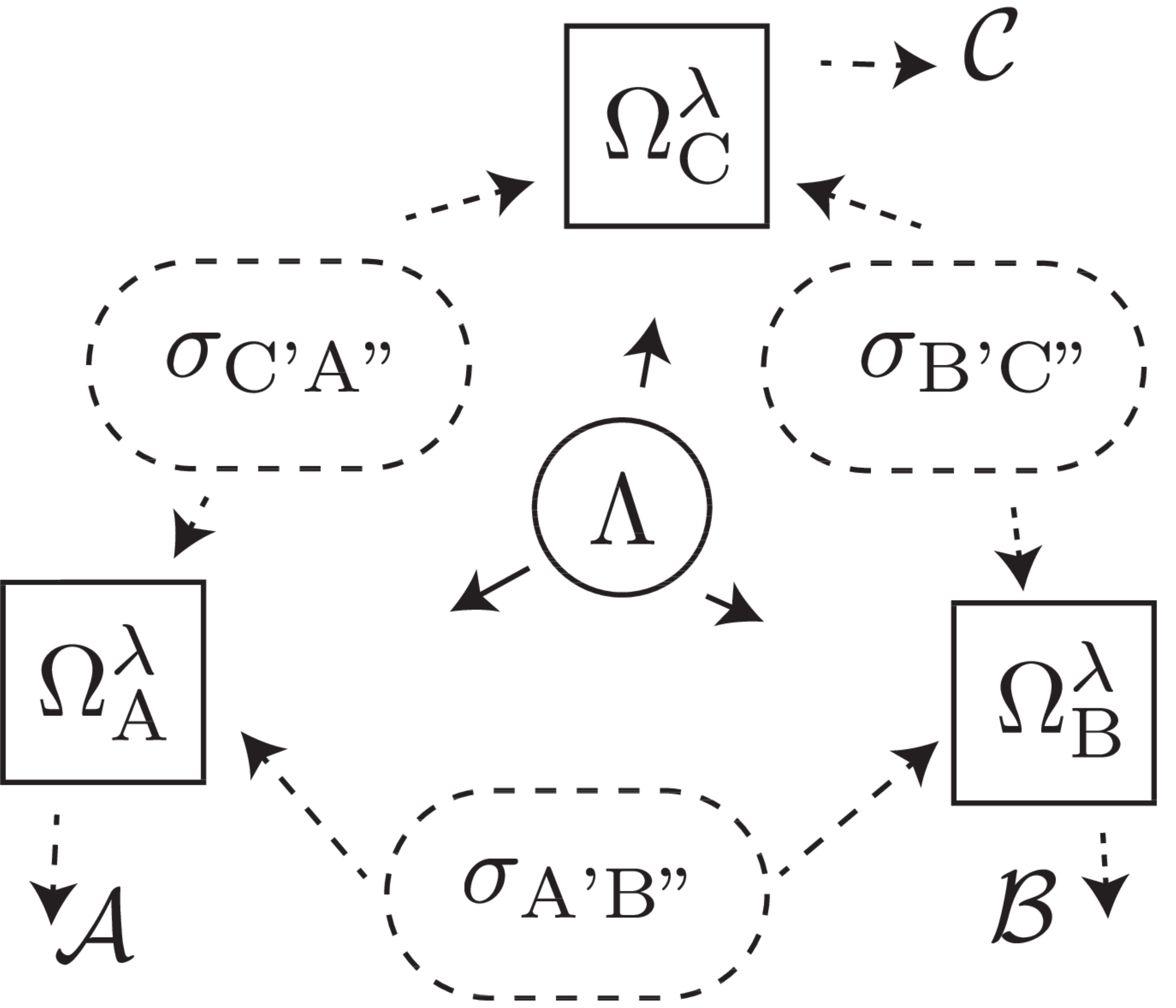}
  \caption{\label{Fig:Network}Network producing a nongenuine network
  $3$-entangled state; quantum resources and spaces are denoted using dotted
  lines, while classical variables are drawn using solid lines.}
\end{figure}

By performing local tomography on the state $\rho_{\mathcal{ABC}}$, can Alice, Bob and Charlie certify that the state produced by Eve's network is indeed a valuable tripartite quantum state?

The family of states that they try to rule out can be defined formally.
Let $\Lambda$ be a classical random variable with distribution $P_{\Lambda} (\lambda)$ sent to the three labs (for example through radio broadcast). Denoting by $\mathsf{B} (\mathcal{H})$ the set of bounded operators on the Hilbert space $\mathcal{H}$, we describe the deterministic operation at Alice's by a family of linear maps $\left\{ \Omega_{\mathcal{A}}^{\lambda} \right\}_{\lambda}$, where each $\Omega_{\mathcal{A}}^{\lambda}$ has type \[ \Omega_{\mathcal{A}}^{\lambda} : \mathsf{B} (\mathcal{A}' \otimes \mathcal{A}'') \rightarrow \mathsf{B} (\mathcal{A}) \] and each $\Omega_{\mathcal{A}}^{\lambda}$ is completely positive and trace preserving.
For completeness, the other maps correspond to \mbox{$\Omega_{\mathcal{B}}^{\lambda} : \mathsf{B} (\mathcal{B}' \otimes \mathcal{B}'') \rightarrow \mathsf{B} (\mathcal{B})$} and $\Omega_{\mathcal{C}}^{\lambda} : \mathsf{B} (\mathcal{C}' \otimes \mathcal{C}'') \rightarrow \mathsf{B} (\mathcal{C})$, so that the state $\rho_{\mathcal{ABC}}$ is
\begin{equation}
  \label{Eq:NotGNE} \rho_{\mathcal{ABC}} \,=\, \sum P_{\Lambda} (\lambda) \left[
  \Omega_{\mathcal{A}}^{\lambda} \otimes \Omega_{\mathcal{B}}^{\lambda} \otimes
  \Omega_{\mathcal{C}}^{\lambda} \right] \left( \sigma \right),
\end{equation}
where $\sigma = \sigma_{\mathcal{A'B''}} \otimes \sigma_{\mathcal{B'C''}} \otimes \sigma_{\mathcal{C'A''}}$.

The valuable states, those \emph{genuinely network 3-entangled}, are those that cannot be written the way described by Eq.~\eqref{Eq:NotGNE}. It is easy to see that the set of states of the form of Eq.~\eqref{Eq:NotGNE} is closed under tensor products and LOSR transformations. That is, the set of network $2$-entangled states is a self-contained class within the resource theory of LOSR entanglement~\cite{Buscemi2012,schmid2019typeindependent}. This property has obvious implications for the monotonicity of any network $3$-entanglement measure. Think for instance of the robustness of entanglement \cite{robustness}. We could define its network $3$-entanglement generalization as the minimum amount of network $2$-entangled noise $R(\rho_{\A\B\C})$ which one must add to a tripartite quantum state $\rho_{\A\B\C}$ to make it network $2$-entangled. Closure under LOSR implies that $R(\rho_{\A\B\C})$ is monotonically decreasing under LOSR operations. From our motivating discussion, though, it follows that $R(\rho_{\A\B\C})$ can be arbitrarily increased by means of LOCC protocols.

Note that, in the considered adversarial scenario, rather than the state $\sigma_{\A'\B''}\otimes\sigma_{\C'\A''}\otimes\sigma_{\B'\C''}$, Eve could also distribute Alice, Bob and Charlie arbitrary convex combinations of states of the form $\sigma^{(i)}_{\A'\B''}\otimes\sigma^{(i)}_{\C'\A''}\otimes\sigma^{(i)}_{\B'\C''}$, for some values of $i$. Since the dimensionality of the primed spaces is unbounded, though, this strategy can be simulated with the operations allowed by Eq.~\eqref{Eq:NotGNE}. Indeed, it suffices to  distribute the tensor product of the states $\sigma^{(i)}_{\A'\B''}\otimes\sigma^{(i)}_{\C'\A''}\otimes\sigma^{(i)}_{\B'\C''}$ and embed the index $i$ within the hidden variable $\Lambda$ (whose dimension is also unbounded). The index $i$ would then signal in which pair of Hilbert spaces at party $Z$'s the map $\Omega_\mathcal{Z}^\lambda$ is to be applied.

The definition of genuine network entanglement can be straightforwardly extended to the $n$-partite case.

\begin{definition}
A multipartite quantum state is {\tmem{genuinely network $k$-entangled}} if it cannot be generated by distributing entangled states among subsets of maximum ${k{\shortminus}1}$ parties, and letting the parties apply local trace-preserving maps, those maps being possibly correlated through global shared randomness.
\end{definition}

\paragraph{Witnesses of genuine network entanglement.}
The certification of $\rho_{\mathcal{ABC}}$ being genuinely network \mbox{$3$-entangled} is complicated, as the dimensions of the Hilbert spaces $\mathcal{A}', \ldots,\,\mathcal{C}''$ are in principle unbounded.
To classify the degree of a state's network multipartiteness, we must somehow determine if the state can come about from a particular quantum causal process.
The study of quantum causal processes has experienced great progress~\cite{Chaves2015,Costa2016,quantNetwork,qinflation,Barrett2019QCM}, and many techniques have recently been developed~\cite{qinflation,Pozas2019,Bowles2019}. Herein, we adapt the inflation technique for causal inference~{\cite{inflation,qinflation}} in order derive witnesses for genuine network entanglement.

As a starter, we consider a three qudit state $\rho_{\A\B\C}$, and quantify its proximity to the Greenberger-Horne-Zeilinger (\textsf{GHZ}) state~{\cite{GHZ}} via the fidelity
\begin{align}
F_{\mathsf{GHZ}_d}&\equiv \bra{\mathsf{GHZ}_d}\rho_{\A\B\C}\ket{\mathsf{GHZ}_d},\\\nonumber
\text{where}\quad\ket{\mathsf{GHZ}_d }&=\sum_{i = 1}^d \frac{\ket{ i i i
  }}{ \sqrt{d}}.
\end{align}

If $\rho_{\mathcal{ABC}}$ is of the form of Eq.~\eqref{Eq:NotGNE}, then there exists a random variable $\Lambda$, quantum states $\sigma_{\mathcal{A'B''}}$, $\sigma_{\mathcal{B'C''}}$ and $\sigma_{\mathcal{C'A''}}$ and families of completely positive and trace-preserving (CPTP) maps $\left\{ \Omega_{\mathcal{A}}^{\lambda} \right\}_{\lambda}$, $\left\{ \Omega_{\mathcal{B}}^{\lambda} \right\}_{\lambda}$ and $\left\{ \Omega_{\mathcal{C}}^{\lambda} \right\}_{\lambda}$ that generate $\rho_{\mathcal{ABC}}$.
To derive bounds on the maximum fidelity achievable by network $2$-entangled states, we next imagine what states one could prepare by combining multiple realizations of the above state and channel resources. As we will see, some of the reduced density matrices of the resulting many-body \emph{inflated states} are fully determined by the original tripartite state $\rho_{\A\B\C}$. The property of $\rho_{\A\B\C}$ admitting a decomposition of the form of Eq.~\eqref{Eq:NotGNE} will then be relaxed to that of admitting positive semidefinite inflated states satisfying said linear constraints. In the language of \cite{inflation}, we will be defining a nonfanout inflation of the causal scenario depicted in Figure \ref{Fig:Network}.

In this regard, consider the \emph{ring inflation} scenario depicted in Figure~\ref{fig:Inflation}. If one acts on two copies of the states $\sigma_{\mathcal{A'B''}}$, $\sigma_{\mathcal{B'C''}}$ and $\sigma_{\mathcal{C'A''}}$ with the maps $\left\{ \Omega_{\mathcal{A}}^{\lambda} \right\}_{\lambda}$, $\left\{ \Omega_{\mathcal{B}}^{\lambda} \right\}_{\lambda}$ and $\left\{ \Omega_{\mathcal{C}}^{\lambda} \right\}_{\lambda}$ in the ways indicated in the figure, one obtains the six-partite density matrices
$\tau_{\A_1 \B_1 \C_1 \A_2 \B_2 \C_2}$ and
$\gamma_{\A_3 \B_3 \C_3 \A_4 \B_4 \C_4}$.
Those are essentially unknown to us, as we do not know how Eve's devices act
when they are wired differently.

\begin{figure}[t]
  \includegraphics{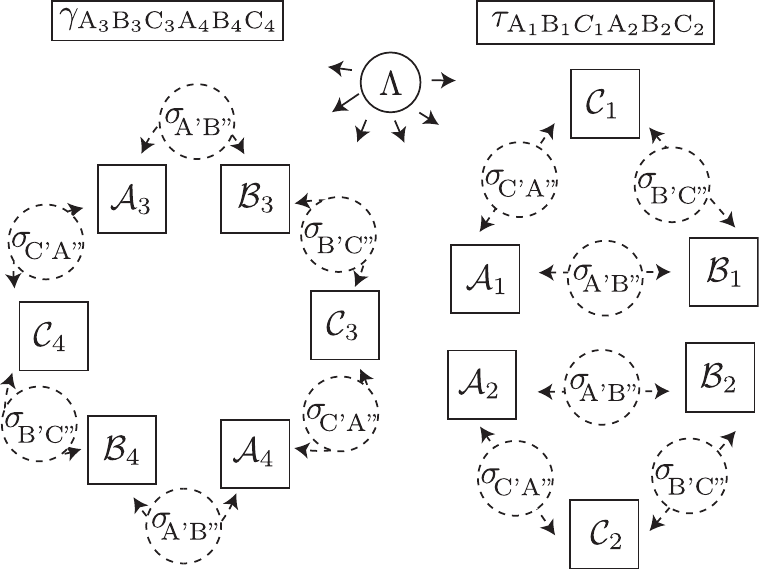}
  \caption{\label{fig:Inflation}Ring inflation of the triangle scenario in Figure~\ref{Fig:Network}, containing copies of the state processing devices $\Omega_{\mathcal{A}, \mathcal{B}, \mathcal{C}}^{\lambda}$; we label such copies according to their output Hilbert space $\mathcal{A}_i, \mathcal{B}_j, \mathcal{C}_k$, where $i, j, k$ is the index of the copy. These devices process copies of the quantum resources $\sigma_{\mathcal{A'B''}}$, $\sigma_{\mathcal{B'C''}}$ and  $\sigma_{\mathcal{C'A''}}$. To simplify the drawing, we omitted the indices
  of these copies and only indicate their original type. Note that, despite the fact that the \emph{wirings} between states and CPTP maps are different than in the original scenario, every copy of a CPTP map acts on copies of the states determined by the original scenario. }
\end{figure}

However, the states $\tau$ and $\gamma$ are subject to several consistency constraints. To begin, with, $\tau$ is symmetric under the exchange of systems ${\A_1\B_1\C_1}$ by systems ${\A_2\B_2\C_2}$, and so is $\gamma$ under the exchange of ${\A_3\B_3\C_3}$ by ${\A_4\B_4\C_4}$. In addition, we observe that
\begin{equation}
  \label{Eq:TwoTriangles} \tau_{\left( \A_1 \B_1 \C_1 \right)} =
  \tau_{\left( \A_2 \B_2 \C_2 \right)} = \rho_{\A\B\C} .
\end{equation}
\noindent Still, we cannot say that $\tau_{\A_1 \B_1 \C_1 \A_2
\B_2 \C_2} = \rho_{\A\B\C} \otimes \rho_{\A\B\C}$ as the
production of the two triangles could be classically correlated through the
shared randomness $\Lambda$. However, the state $\tau$ is separable across the
$\A_1 \B_1 \C_1 / \A_2 \B_2 \C_2$ partition. Both $\gamma$ and $\tau$ are related to each other through the constraints
\begin{equation}
  \label{Eq:Ring} \gamma_{\left( \A_3 \B_3 \A_4 \B_4
  \right)} = \tau_{\left( \A_1 \B_1 \A_2 \B_2 \right)}
\end{equation}
and $\gamma_{\left( \B_3 \C_3 \B_4 \C_4 \right)} =
\tau_{\left( \B_1 \C_1 \B_2 \C_2 \right)}$ and
$\gamma_{\left( \C_3 \A_4 \C_4 \A_3 \right)} =
\tau_{\left( \C_1 \A_1 \C_2 \A_2 \right)}$. Furthermore, $\tau$ and $\gamma$ have trace one and are semidefinite positive.
Finally, the reduced state $\gamma_{\left( \A_3 \B_3 \C_3
\B_4 \right)}$ is separable across the $\A_3 \B_3 \C_3
/ \B_4$ partition; and additional constraints of that type follow from cyclic symmetry.

Let us now provide some intuition as to why any state $\rho_{\A\B\C}$ admitting such extensions ${\tau,\gamma}$ cannot be arbitrarily close to the \textsf{GHZ} state. Suppose, indeed, that $F_{\mathsf{GHZ}_d}=1$, i.e., $\rho_{\A\B\C}= \proj{\mathsf{GHZ}_d}$ and that there exist extensions ${\gamma,\tau}$ satisfying the constraints above. A measurement in the computational basis of the sites ${\A_3,\B_3,\C_3}$ of $\gamma$ will generate the random variables ${a_3, b_3, c_3}$. Since $\gamma_{(\A_3\B_3)}=\rho_{(\A\B)}=\frac{1}{d}\sum_{i=1}^d\proj{i,i}$, it must be the case that ${a_3, b_3}$ are perfectly correlated. The same considerations hold for $b_3$ and $c_3$. Since ${a_3,b_3}$ and ${b_3,c_3}$ are pair-wise perfectly correlated, so are ${a_3,c_3}$. Now, from the condition $\gamma_{\left(\C_3 \A_4 \C_4 \A_3 \right)} =
\tau_{\left(\C_2 \A_2 \C_1 \A_1 \right)}$, we have that the distribution of $c_3$ and $a_3$ must be the same as that of $c_2$ and $a_1$. Hence, $c_2$ and $a_1$ must be perfectly correlated. However, $\tau_{(\A_1\B_1\C_1)}$ is a pure state, since $\tau_{(\A_1\B_1\C_1)}=\rho_{\A\B\C}=\proj{\mathsf{GHZ}_d}$, and hence it must be in a product state with respect to any other system, such as $\C_2$. It follows that a measurement in the computational basis of the sites $\A_1$ and $\C_2$ will produce two uncorrelated random variables ${a_1,c_2}$. We thus reach a contradiction.

The previous argument just invalidates the case $F_{\mathsf{GHZ}_d}=1$. A more elaborate argument (see Appendix A for a proof) shows that, if ${a,b,c}$ are the random variables resulting from measuring $\rho_{\A\B\C}$ locally, then any network 2-entangled state $\rho_{\A\B\C}$ must satisfy
\begin{align}\begin{split}
H(a:b) + H(b:c)&-H(b)\leq\\ &S(\rho_{(\A)})+S(\rho_{\A\B\C})-S(\rho_{(\B\C)}).
\label{neumann}
\end{split}\end{align}
\noindent Here $H(x)$, $H(x:y)$ and $S(\rho)$ respectively denote the Shannon entropy of variable $x$, the mutual information between the random variables $x,y$, and the von Neumann entropy of state $\rho$. Condition (\ref{neumann}) is clearly violated if $\rho_{\A\B\C}\approx\proj{\mathsf{GHZ}_d}$ and the measurements are carried in the computational basis.

Another constraint satisfied by states satisfying Eq.~\eqref{Eq:NotGNE}, expressed in terms of the \textsf{GHZ} fidelity, is
\be
F_{\mathsf{GHZ}_d}\leq \frac{2d (3 d + \sqrt{2d-1})}{1-2 d + 9 d^2}.
\label{fidelius}
\ee
\noindent Remarkably, in order to derive Eqs. \eqref{neumann} and \eqref{fidelius}, it is not necessary to invoke the existence of the six-partite states $\tau, \gamma$, but that of their reduced density matrices $\tau_{(\A_1\B_1\C_1\C_2)}$, $\gamma_{(\A_3\B_3\C_3)}$. As shown in Appendix B, both expressions, Eqs.~\eqref{neumann} and \eqref{fidelius}, can be generalized to detect genuine network $k$-entanglement.

For $d{=}2$, Eq.~\eqref{fidelius} establishes that any tripartite state with $F_{\mathsf{GHZ}_2}>\frac{4}{33} \left(6+\sqrt{3}\right)\approx 0.9372$ is genuinely network $3$-entangled. As it turns out, this inequality is not tight: it can be improved to $F_{\mathsf{GHZ}_2}>\frac{1+\sqrt{3}}{4}\approx 0.6803$ by means of semidefinite programming applied to the ring inflation.

The variables in the corresponding program are trace-one positive semidefinite matrices $\tau_{\A_1 \B_1 \C_1 \A_2 \B_2
\C_2}$ and $\gamma_{\A_3 \B_3 \C_3 \A_4
\B_4 \C_4}$ of size $64 \times 64$, subject to linear constraints of the form of Eqs.~\eqref{Eq:TwoTriangles} and \eqref{Eq:Ring}, as well as to the permutational symmetry $1\leftrightarrow 2$, $3\leftrightarrow 4$. For all states $\mu_{\X\Y}$ separable across a $\X/\Y$
partition, we add a Positivity under Partial Transposition (PPT) constraint ${\left( \mu_{\X\Y}
\right)^{_{\top_{\Y}}} \!\!\!\underaccent{\mathsf{ PSD}}{\succeq} 0}$~\cite{PPT}. This applies to $\tau$
across the $\A_1 \B_1 \C_1 / \A_2 \B_2 \C_2$
partition, and to reduced states of $\gamma$ for the partitions $\A_3
\B_3 \C_3 / \B_4$, $\B_3 \C_3 \A_4 /
\C_4$, $\C_3 \A_4 \B_4 / \A_3$.

The bound \mbox{$F_{\mathsf{GHZ}_2}>\frac{1+\sqrt{3}}{4}$} is obtained by maximizing $\left\langle \mathsf{GHZ}_2 \middle| \rho_{\mathcal{ABC}} \middle| \mathsf{GHZ}_2 \right\rangle$ subject to the constraints above---a typical instance of a semidefinite program---using the optimization toolbox CVX~{\cite{cvx}} and the solver \textsc{Mosek}~{\cite{mosek}}.

We also employed the semidefinite optimization procedure using as reference the \textsf{W} state~{\cite{W_state}}, ${\ket{\mathsf{W}} \equiv\, \displaystyle \frac{\ket{001} + \ket{010} + \ket{100}}{\sqrt{3}}}$, concluding that any $3$-qubit state $\rho_{\mathcal{ABC}}$ with $\bra{\mathsf{W}} \rho_{\mathcal{ABC}} \ket{\mathsf{W}} > 0.7602$ is genuinely network $3$-entangled.

Armed with these witnesses, we find that several past experiments in quantum optics can be interpreted as demonstrations of genuine network tripartite entanglement~\cite{exp1,exp2,exp3,exp4}.
Indeed, in all those experiments, the fidelity of the prepared states with respect to \textsf{GHZ} or \textsf{W} states is greater than the bounds derived above for network-bipartite states. The prepared states are thus certified to contain genuine network tripartite entanglement.

\paragraph{Robustness to detection inefficiency.}
In many experimental setups, due to low detector efficiencies, the carriers transmitting the quantum information are often unobserved.
The standard prescription in such a predicament consists in discarding the experimental data gathered when not all detectors click. Coming back to our adversarial setup, this postselection of measurement results opens a loophole that Eve can in principle exploit to fool Alice, Bob, and Charlie.
It is possible to contemplate this contingency in the calculations above, and thus bound the detection efficiency needed for certifying genuine network entanglement under post-selection.

Let $p$ indicate the fraction of experimental data preserved by postselection, i.e., the probability that all three detectors click. If $\rho^p_\mathcal{ABC}$ is the state reconstructed after postselection, then all that can be said about the true tripartite quantum state $\rho_\mathcal{ABC}$ before the postselection took place is that
\be\label{eq:postselection}
  \rho_\mathcal{ABC}-p\times \rho^p_\mathcal{ABC}\succeq 0.
\ee
As before, linear optimizations over the set of postselected states $\rho^p_\mathcal{ABC}$ can be conducted via semidefinite programming.
In such instances one continues to relate the inflated states $\tau$ and $\gamma$ to the true (albeit unknown) tripartite state $\rho_\mathcal{ABC}$, and Eq.~\eqref{eq:postselection} is merely added as an extra constraint.
We find critical postselection probabilities beyond which one can still certify genuine network tripartite entanglement via \textsf{GHZ} fidelity ($p_c\approx 0.685$) or \textsf{W} fidelity ($p_c\approx 0.765$).

\paragraph{Conclusions.}
In this paper we have argued that the standard definition of genuine multipartite entanglement is not appropriate to assess the quantum control over an ever-growing number of quantum systems.
We proposed an alternative definition, \emph{genuine network multipartite entanglement}, that captures the potential of a source to distribute entanglement over a number of spatially separated parties.
We provided analytic and numerical tools to detect genuine network tripartite entanglement, and also indicated how the definition can be adapted to situations where there may be local postselections on each party's lab. Furthermore, the construction can be adapted to detect genuine network $n$-partite entanglement for any $n$.

While quite general, our numerical methods to detect genuine network entanglement demand considerable memory resources to the point that we were not able to derive new entanglement witnesses for tripartite qutrit states in a normal computer.
In addition, there exist significant gaps between the bounds we derived on \textsf{GHZ} and \textsf{W} state fidelities via SDP relaxations and the lower bounds obtained using standard variational techniques~\cite{seeSaw1,seeSaw2}.
Using such algorithms, we were not able to give lower bounds to the \textsf{GHZ} and \textsf{W} fidelities larger than $0.5170$ and $2/3$, respectively.
A topic for future research is thus to develop better techniques for the characterization of genuine network multipartite entanglement.

\paragraph{Note added.} After completing this manuscript, we became aware of the work of \cite{kraft2020quantum,luo}, whose authors consider a scenario very similar to that depicted in Figure \ref{Fig:Network}. Crucially, they restrict the maps $\Omega^\lambda_{\A,\B,\C}$ to be unitary transformations, acting on convex combinations of bipartite states. The restriction to unitary maps not only allows upper-bounding the dimensionality of the source states $\sigma_{\A'\B''},\sigma_{\C'\A''},\sigma_{\B'\C''}$, but it also severely constrains the resulting set of states $\Delta_C$: as shown in \cite{kraft2020quantum}, tripartite qubit states in $\Delta_C$ cannot be genuinely tripartite entangled. This contrasts with the \textsf{GHZ} fidelity greater than $1/2$ reported above, achievable by states of the form of Eq.~\eqref{Eq:NotGNE}.

\begin{acknowledgments}\paragraph{Acknowledgments.}
We thank Antonio Ac\'in, Jean-Daniel Bancal, T.C.~Fraser, Yeong-Cherng Liang, David Schmid, and Robert~W.~Spekkens for useful discussions.
M.N. was supported by the Austrian Science fund (FWF) stand-alone project P 30947.
The work of A.P.-K. was supported by Fundaci\'o Obra Social ``la Caixa'' (LCF/BQ/ES15/10360001), the ERC (CoG QITBOX and the European Union's Horizon 2020 research and innovation programme - grant agreement No 648913), the Spanish MINECO (FIS2016-80773-P and Severo Ochoa SEV-2015-0522), Fundaci\'o Cellex, and Generalitat de Catalunya (SGR 1381 and CERCA Programme).
This research was supported by Perimeter Institute for Theoretical Physics. Research at Perimeter Institute is supported in part by the Government of Canada through the Department of Innovation, Science and Economic Development Canada and by the Province of Ontario through the Ministry of Economic Development, Job Creation and Trade.
This publication was made possible through the support of a grant from the John Templeton Foundation. The opinions expressed in this publication are those of the authors and do not necessarily reflect the views of the John Templeton Foundation.
\end{acknowledgments}

\bibliographystyle{apsrev4-2}
\nocite{apsrev41Control}
\bibliography{entanglement2}

\begin{thebibliography}{39}%
\makeatletter
\providecommand \@ifxundefined [1]{%
 \@ifx{#1\undefined}
}%
\providecommand \@ifnum [1]{%
 \ifnum #1\expandafter \@firstoftwo
 \else \expandafter \@secondoftwo
 \fi
}%
\providecommand \@ifx [1]{%
 \ifx #1\expandafter \@firstoftwo
 \else \expandafter \@secondoftwo
 \fi
}%
\providecommand \natexlab [1]{#1}%
\providecommand \enquote  [1]{``#1''}%
\providecommand \bibnamefont  [1]{#1}%
\providecommand \bibfnamefont [1]{#1}%
\providecommand \citenamefont [1]{#1}%
\providecommand \href@noop [0]{\@secondoftwo}%
\providecommand \href [0]{\begingroup \@sanitize@url \@href}%
\providecommand \@href[1]{\@@startlink{#1}\@@href}%
\providecommand \@@href[1]{\endgroup#1\@@endlink}%
\providecommand \@sanitize@url [0]{\catcode `\\12\catcode `\$12\catcode
  `\&12\catcode `\#12\catcode `\^12\catcode `\_12\catcode `\%12\relax}%
\providecommand \@@startlink[1]{}%
\providecommand \@@endlink[0]{}%
\providecommand \url  [0]{\begingroup\@sanitize@url \@url }%
\providecommand \@url [1]{\endgroup\@href {#1}{\urlprefix }}%
\providecommand \urlprefix  [0]{URL }%
\providecommand \Eprint [0]{\href }%
\providecommand \doibase [0]{https://doi.org/}%
\providecommand \selectlanguage [0]{\@gobble}%
\providecommand \bibinfo  [0]{\@secondoftwo}%
\providecommand \bibfield  [0]{\@secondoftwo}%
\providecommand \translation [1]{[#1]}%
\providecommand \BibitemOpen [0]{}%
\providecommand \bibitemStop [0]{}%
\providecommand \bibitemNoStop [0]{.\EOS\space}%
\providecommand \EOS [0]{\spacefactor3000\relax}%
\providecommand \BibitemShut  [1]{\csname bibitem#1\endcsname}%
\let\auto@bib@innerbib\@empty
\bibitem [{\citenamefont {Bennett}\ \emph {et~al.}(1993)\citenamefont
  {Bennett}, \citenamefont {Brassard}, \citenamefont {Cr\'epeau}, \citenamefont
  {Jozsa}, \citenamefont {Peres},\ and\ \citenamefont
  {Wootters}}]{teleportation}%
  \BibitemOpen
  \bibfield  {author} {\bibinfo {author} {\bibfnamefont {C.~H.}\ \bibnamefont
  {Bennett}}, \bibinfo {author} {\bibfnamefont {G.}~\bibnamefont {Brassard}},
  \bibinfo {author} {\bibfnamefont {C.}~\bibnamefont {Cr\'epeau}}, \bibinfo
  {author} {\bibfnamefont {R.}~\bibnamefont {Jozsa}}, \bibinfo {author}
  {\bibfnamefont {A.}~\bibnamefont {Peres}},\ and\ \bibinfo {author}
  {\bibfnamefont {W.~K.}\ \bibnamefont {Wootters}},\ }\bibfield  {title}
  {\bibinfo {title} {{Teleporting an unknown quantum state via dual classical
  and Einstein-Podolsky-Rosen channels}},\ }\href
  {https://doi.org/10.1103/PhysRevLett.70.1895} {\bibfield  {journal} {\bibinfo
   {journal} {Phys. Rev. Lett.}\ }\textbf {\bibinfo {volume} {70}},\ \bibinfo
  {pages} {1895} (\bibinfo {year} {1993})}\BibitemShut {NoStop}%
\bibitem [{\citenamefont {Bennett}\ and\ \citenamefont
  {Wiesner}(1992)}]{denseCoding}%
  \BibitemOpen
  \bibfield  {author} {\bibinfo {author} {\bibfnamefont {C.~H.}\ \bibnamefont
  {Bennett}}\ and\ \bibinfo {author} {\bibfnamefont {S.~J.}\ \bibnamefont
  {Wiesner}},\ }\bibfield  {title} {\bibinfo {title} {{Communication via one-
  and two-particle operators on Einstein-Podolsky-Rosen states}},\ }\href
  {https://doi.org/10.1103/PhysRevLett.69.2881} {\bibfield  {journal} {\bibinfo
   {journal} {Phys. Rev. Lett.}\ }\textbf {\bibinfo {volume} {69}},\ \bibinfo
  {pages} {2881} (\bibinfo {year} {1992})}\BibitemShut {NoStop}%
\bibitem [{\citenamefont {Scarani}\ \emph {et~al.}(2009)\citenamefont
  {Scarani}, \citenamefont {Bechmann-Pasquinucci}, \citenamefont {Cerf},
  \citenamefont {Du\ifmmode~\check{s}\else \v{s}\fi{}ek}, \citenamefont
  {L\"utkenhaus},\ and\ \citenamefont {Peev}}]{qkd}%
  \BibitemOpen
  \bibfield  {author} {\bibinfo {author} {\bibfnamefont {V.}~\bibnamefont
  {Scarani}}, \bibinfo {author} {\bibfnamefont {H.}~\bibnamefont
  {Bechmann-Pasquinucci}}, \bibinfo {author} {\bibfnamefont {N.~J.}\
  \bibnamefont {Cerf}}, \bibinfo {author} {\bibfnamefont {M.}~\bibnamefont
  {Du\ifmmode~\check{s}\else \v{s}\fi{}ek}}, \bibinfo {author} {\bibfnamefont
  {N.}~\bibnamefont {L\"utkenhaus}},\ and\ \bibinfo {author} {\bibfnamefont
  {M.}~\bibnamefont {Peev}},\ }\bibfield  {title} {\bibinfo {title} {The
  security of practical quantum key distribution},\ }\href
  {https://doi.org/10.1103/RevModPhys.81.1301} {\bibfield  {journal} {\bibinfo
  {journal} {Rev. Mod. Phys.}\ }\textbf {\bibinfo {volume} {81}},\ \bibinfo
  {pages} {1301} (\bibinfo {year} {2009})}\BibitemShut {NoStop}%
\bibitem [{\citenamefont {Bell}(1964)}]{bell}%
  \BibitemOpen
  \bibfield  {author} {\bibinfo {author} {\bibfnamefont {J.~S.}\ \bibnamefont
  {Bell}},\ }\bibfield  {title} {\bibinfo {title} {{On the
  Einstein-Podolsky-Rosen paradox}},\ }\href
  {https://doi.org/10.1103/PhysicsPhysiqueFizika.1.195} {\bibfield  {journal}
  {\bibinfo  {journal} {Physics}\ }\textbf {\bibinfo {volume} {1}},\ \bibinfo
  {pages} {195} (\bibinfo {year} {1964})}\BibitemShut {NoStop}%
\bibitem [{\citenamefont {Brunner}\ \emph {et~al.}(2014)\citenamefont
  {Brunner}, \citenamefont {Cavalcanti}, \citenamefont {Pironio}, \citenamefont
  {Scarani},\ and\ \citenamefont {Wehner}}]{BrunnerRMP}%
  \BibitemOpen
  \bibfield  {author} {\bibinfo {author} {\bibfnamefont {N.}~\bibnamefont
  {Brunner}}, \bibinfo {author} {\bibfnamefont {D.}~\bibnamefont {Cavalcanti}},
  \bibinfo {author} {\bibfnamefont {S.}~\bibnamefont {Pironio}}, \bibinfo
  {author} {\bibfnamefont {V.}~\bibnamefont {Scarani}},\ and\ \bibinfo {author}
  {\bibfnamefont {S.}~\bibnamefont {Wehner}},\ }\bibfield  {title} {\bibinfo
  {title} {Bell nonlocality},\ }\href
  {https://link.aps.org/doi/10.1103/RevModPhys.86.419} {\bibfield  {journal}
  {\bibinfo  {journal} {Rev. Mod. Phys.}\ }\textbf {\bibinfo {volume} {86}},\
  \bibinfo {pages} {419} (\bibinfo {year} {2014})}\BibitemShut {NoStop}%
\bibitem [{\citenamefont {Lu}\ \emph {et~al.}(2007)\citenamefont {Lu},
  \citenamefont {Zhou}, \citenamefont {G\"uhne}, \citenamefont {Gao},
  \citenamefont {Zhang}, \citenamefont {Yuan}, \citenamefont {Goebel},
  \citenamefont {Yang},\ and\ \citenamefont {Pan}}]{sixPhotons}%
  \BibitemOpen
  \bibfield  {author} {\bibinfo {author} {\bibfnamefont {C.-Y.}\ \bibnamefont
  {Lu}}, \bibinfo {author} {\bibfnamefont {X.-Q.}\ \bibnamefont {Zhou}},
  \bibinfo {author} {\bibfnamefont {O.}~\bibnamefont {G\"uhne}}, \bibinfo
  {author} {\bibfnamefont {W.-B.}\ \bibnamefont {Gao}}, \bibinfo {author}
  {\bibfnamefont {J.}~\bibnamefont {Zhang}}, \bibinfo {author} {\bibfnamefont
  {Z.-S.}\ \bibnamefont {Yuan}}, \bibinfo {author} {\bibfnamefont
  {A.}~\bibnamefont {Goebel}}, \bibinfo {author} {\bibfnamefont
  {T.}~\bibnamefont {Yang}},\ and\ \bibinfo {author} {\bibfnamefont {J.-W.}\
  \bibnamefont {Pan}},\ }\bibfield  {title} {\bibinfo {title} {{Experimental
  entanglement of six photons in graph states}},\ }\href
  {https://doi.org/10.1038/nphys507} {\bibfield  {journal} {\bibinfo  {journal}
  {Nat. Phys.}\ }\textbf {\bibinfo {volume} {3}},\ \bibinfo {pages} {91}
  (\bibinfo {year} {2007})}\BibitemShut {NoStop}%
\bibitem [{\citenamefont {Yao}\ \emph {et~al.}(2012)\citenamefont {Yao},
  \citenamefont {Wang}, \citenamefont {Xu}, \citenamefont {Lu}, \citenamefont
  {Pan}, \citenamefont {Bao}, \citenamefont {Peng}, \citenamefont {Lu},
  \citenamefont {Chen},\ and\ \citenamefont {Pan}}]{eightPhotons}%
  \BibitemOpen
  \bibfield  {author} {\bibinfo {author} {\bibfnamefont {X.-C.}\ \bibnamefont
  {Yao}}, \bibinfo {author} {\bibfnamefont {T.-X.}\ \bibnamefont {Wang}},
  \bibinfo {author} {\bibfnamefont {P.}~\bibnamefont {Xu}}, \bibinfo {author}
  {\bibfnamefont {H.}~\bibnamefont {Lu}}, \bibinfo {author} {\bibfnamefont
  {G.-S.}\ \bibnamefont {Pan}}, \bibinfo {author} {\bibfnamefont {X.-H.}\
  \bibnamefont {Bao}}, \bibinfo {author} {\bibfnamefont {C.-Z.}\ \bibnamefont
  {Peng}}, \bibinfo {author} {\bibfnamefont {C.-Y.}\ \bibnamefont {Lu}},
  \bibinfo {author} {\bibfnamefont {Y.-A.}\ \bibnamefont {Chen}},\ and\
  \bibinfo {author} {\bibfnamefont {J.-W.}\ \bibnamefont {Pan}},\ }\bibfield
  {title} {\bibinfo {title} {{Observation of eight-photon entanglement}},\
  }\href {https://doi.org/10.1038/nphoton.2011.354} {\bibfield  {journal}
  {\bibinfo  {journal} {Nat. Phot.}\ }\textbf {\bibinfo {volume} {6}},\
  \bibinfo {pages} {225} (\bibinfo {year} {2012})}\BibitemShut {NoStop}%
\bibitem [{\citenamefont {Wang}\ \emph {et~al.}(2016)\citenamefont {Wang},
  \citenamefont {Chen}, \citenamefont {Li}, \citenamefont {Huang},
  \citenamefont {Liu}, \citenamefont {Chen}, \citenamefont {Luo}, \citenamefont
  {Su}, \citenamefont {Wu}, \citenamefont {Li}, \citenamefont {Lu},
  \citenamefont {Hu}, \citenamefont {Jiang}, \citenamefont {Peng},
  \citenamefont {Li}, \citenamefont {Liu}, \citenamefont {Chen}, \citenamefont
  {Lu},\ and\ \citenamefont {Pan}}]{tenPhotons}%
  \BibitemOpen
  \bibfield  {author} {\bibinfo {author} {\bibfnamefont {X.-L.}\ \bibnamefont
  {Wang}}, \bibinfo {author} {\bibfnamefont {L.-K.}\ \bibnamefont {Chen}},
  \bibinfo {author} {\bibfnamefont {W.}~\bibnamefont {Li}}, \bibinfo {author}
  {\bibfnamefont {H.-L.}\ \bibnamefont {Huang}}, \bibinfo {author}
  {\bibfnamefont {C.}~\bibnamefont {Liu}}, \bibinfo {author} {\bibfnamefont
  {C.}~\bibnamefont {Chen}}, \bibinfo {author} {\bibfnamefont {Y.-H.}\
  \bibnamefont {Luo}}, \bibinfo {author} {\bibfnamefont {Z.-E.}\ \bibnamefont
  {Su}}, \bibinfo {author} {\bibfnamefont {D.}~\bibnamefont {Wu}}, \bibinfo
  {author} {\bibfnamefont {Z.-D.}\ \bibnamefont {Li}}, \bibinfo {author}
  {\bibfnamefont {H.}~\bibnamefont {Lu}}, \bibinfo {author} {\bibfnamefont
  {Y.}~\bibnamefont {Hu}}, \bibinfo {author} {\bibfnamefont {X.}~\bibnamefont
  {Jiang}}, \bibinfo {author} {\bibfnamefont {C.-Z.}\ \bibnamefont {Peng}},
  \bibinfo {author} {\bibfnamefont {L.}~\bibnamefont {Li}}, \bibinfo {author}
  {\bibfnamefont {N.-L.}\ \bibnamefont {Liu}}, \bibinfo {author} {\bibfnamefont
  {Y.-A.}\ \bibnamefont {Chen}}, \bibinfo {author} {\bibfnamefont {C.-Y.}\
  \bibnamefont {Lu}},\ and\ \bibinfo {author} {\bibfnamefont {J.-W.}\
  \bibnamefont {Pan}},\ }\bibfield  {title} {\bibinfo {title} {Experimental
  ten-photon entanglement},\ }\href
  {https://doi.org/10.1103/PhysRevLett.117.210502} {\bibfield  {journal}
  {\bibinfo  {journal} {Phys. Rev. Lett.}\ }\textbf {\bibinfo {volume} {117}},\
  \bibinfo {pages} {210502} (\bibinfo {year} {2016})}\BibitemShut {NoStop}%
\bibitem [{\citenamefont {Gross}\ \emph {et~al.}(2010)\citenamefont {Gross},
  \citenamefont {Zibold}, \citenamefont {Nicklas}, \citenamefont {Est\`{e}ve},\
  and\ \citenamefont {Oberthaler}}]{hundredsAtoms}%
  \BibitemOpen
  \bibfield  {author} {\bibinfo {author} {\bibfnamefont {C.}~\bibnamefont
  {Gross}}, \bibinfo {author} {\bibfnamefont {T.}~\bibnamefont {Zibold}},
  \bibinfo {author} {\bibfnamefont {E.}~\bibnamefont {Nicklas}}, \bibinfo
  {author} {\bibfnamefont {J.}~\bibnamefont {Est\`{e}ve}},\ and\ \bibinfo
  {author} {\bibfnamefont {M.~K.}\ \bibnamefont {Oberthaler}},\ }\bibfield
  {title} {\bibinfo {title} {{Nonlinear atom interferometer surpasses classical
  precision limit}},\ }\href {https://doi.org/10.1038/nature08919} {\bibfield
  {journal} {\bibinfo  {journal} {Nature}\ }\textbf {\bibinfo {volume} {464}},\
  \bibinfo {pages} {1165} (\bibinfo {year} {2010})}\BibitemShut {NoStop}%
\bibitem [{\citenamefont {Seevinck}\ and\ \citenamefont
  {Uffink}(2001)}]{wrongDef}%
  \BibitemOpen
  \bibfield  {author} {\bibinfo {author} {\bibfnamefont {M.}~\bibnamefont
  {Seevinck}}\ and\ \bibinfo {author} {\bibfnamefont {J.}~\bibnamefont
  {Uffink}},\ }\bibfield  {title} {\bibinfo {title} {Sufficient conditions for
  three-particle entanglement and their tests in recent experiments},\ }\href
  {https://doi.org/10.1103/PhysRevA.65.012107} {\bibfield  {journal} {\bibinfo
  {journal} {Phys. Rev. A}\ }\textbf {\bibinfo {volume} {65}},\ \bibinfo
  {pages} {012107} (\bibinfo {year} {2001})}\BibitemShut {NoStop}%
\bibitem [{\citenamefont {Gühne}\ \emph {et~al.}(2005)\citenamefont {Gühne},
  \citenamefont {Tóth},\ and\ \citenamefont {Briegel}}]{wrongDef2}%
  \BibitemOpen
  \bibfield  {author} {\bibinfo {author} {\bibfnamefont {O.}~\bibnamefont
  {Gühne}}, \bibinfo {author} {\bibfnamefont {G.}~\bibnamefont {Tóth}},\ and\
  \bibinfo {author} {\bibfnamefont {H.~J.}\ \bibnamefont {Briegel}},\
  }\bibfield  {title} {\bibinfo {title} {Multipartite entanglement in spin
  chains},\ }\href {https://doi.org/10.1088/1367-2630/7/1/229} {\bibfield
  {journal} {\bibinfo  {journal} {New J. Phys.}\ }\textbf {\bibinfo {volume}
  {7}},\ \bibinfo {pages} {229} (\bibinfo {year} {2005})}\BibitemShut {NoStop}%
\bibitem [{\citenamefont {G\"uhne}\ and\ \citenamefont
  {T\'oth}(2006)}]{wrongDef3}%
  \BibitemOpen
  \bibfield  {author} {\bibinfo {author} {\bibfnamefont {O.}~\bibnamefont
  {G\"uhne}}\ and\ \bibinfo {author} {\bibfnamefont {G.}~\bibnamefont
  {T\'oth}},\ }\bibfield  {title} {\bibinfo {title} {Energy and multipartite
  entanglement in multidimensional and frustrated spin models},\ }\href
  {https://doi.org/10.1103/PhysRevA.73.052319} {\bibfield  {journal} {\bibinfo
  {journal} {Phys. Rev. A}\ }\textbf {\bibinfo {volume} {73}},\ \bibinfo
  {pages} {052319} (\bibinfo {year} {2006})}\BibitemShut {NoStop}%
\bibitem [{\citenamefont {McConnell}\ \emph {et~al.}(2015)\citenamefont
  {McConnell}, \citenamefont {Zhang}, \citenamefont {Hu}, \citenamefont {Cuk},\
  and\ \citenamefont {Vuletic}}]{thousandsAtoms}%
  \BibitemOpen
  \bibfield  {author} {\bibinfo {author} {\bibfnamefont {R.}~\bibnamefont
  {McConnell}}, \bibinfo {author} {\bibfnamefont {H.}~\bibnamefont {Zhang}},
  \bibinfo {author} {\bibfnamefont {J.}~\bibnamefont {Hu}}, \bibinfo {author}
  {\bibfnamefont {S.}~\bibnamefont {Cuk}},\ and\ \bibinfo {author}
  {\bibfnamefont {V.}~\bibnamefont {Vuletic}},\ }\bibfield  {title} {\bibinfo
  {title} {{Entanglement with negative Wigner function of almost 3000 atoms
  heralded by one photon}},\ }\href {https://doi.org/10.1038/nature14293}
  {\bibfield  {journal} {\bibinfo  {journal} {Nature}\ }\textbf {\bibinfo
  {volume} {519}},\ \bibinfo {pages} {439} (\bibinfo {year}
  {2015})}\BibitemShut {NoStop}%
\bibitem [{\citenamefont {Allen}\ \emph {et~al.}(2017)\citenamefont {Allen},
  \citenamefont {Barrett}, \citenamefont {Horsman}, \citenamefont {Lee},\ and\
  \citenamefont {Spekkens}}]{quantNetwork}%
  \BibitemOpen
  \bibfield  {author} {\bibinfo {author} {\bibfnamefont {J.-M.~A.}\
  \bibnamefont {Allen}}, \bibinfo {author} {\bibfnamefont {J.}~\bibnamefont
  {Barrett}}, \bibinfo {author} {\bibfnamefont {D.~C.}\ \bibnamefont
  {Horsman}}, \bibinfo {author} {\bibfnamefont {C.~M.}\ \bibnamefont {Lee}},\
  and\ \bibinfo {author} {\bibfnamefont {R.~W.}\ \bibnamefont {Spekkens}},\
  }\bibfield  {title} {\bibinfo {title} {{Quantum Common Causes and Quantum
  Causal Models}},\ }\href {https://doi.org/10.1103/PhysRevX.7.031021}
  {\bibfield  {journal} {\bibinfo  {journal} {Phys. Rev. X}\ }\textbf {\bibinfo
  {volume} {7}},\ \bibinfo {pages} {031021} (\bibinfo {year}
  {2017})}\BibitemShut {NoStop}%
\bibitem [{\citenamefont {Buscemi}(2012)}]{Buscemi2012}%
  \BibitemOpen
  \bibfield  {author} {\bibinfo {author} {\bibfnamefont {F.}~\bibnamefont
  {Buscemi}},\ }\bibfield  {title} {\bibinfo {title} {{All Entangled Quantum
  States Are Nonlocal}},\ }\href
  {https://doi.org/10.1103/PhysRevLett.108.200401} {\bibfield  {journal}
  {\bibinfo  {journal} {Phys. Rev. Lett.}\ }\textbf {\bibinfo {volume} {108}},\
  \bibinfo {pages} {200401} (\bibinfo {year} {2012})}\BibitemShut {NoStop}%
\bibitem [{\citenamefont {Schmid}\ \emph
  {et~al.}(2020{\natexlab{a}})\citenamefont {Schmid}, \citenamefont {Rosset},\
  and\ \citenamefont {Buscemi}}]{schmid2019typeindependent}%
  \BibitemOpen
  \bibfield  {author} {\bibinfo {author} {\bibfnamefont {D.}~\bibnamefont
  {Schmid}}, \bibinfo {author} {\bibfnamefont {D.}~\bibnamefont {Rosset}},\
  and\ \bibinfo {author} {\bibfnamefont {F.}~\bibnamefont {Buscemi}},\
  }\bibfield  {title} {\bibinfo {title} {The type-independent resource theory
  of local operations and shared randomness},\ }\href
  {https://doi.org/10.22331/q-2020-04-30-262} {\bibfield  {journal} {\bibinfo
  {journal} {{Quantum}}\ }\textbf {\bibinfo {volume} {4}},\ \bibinfo {pages}
  {262} (\bibinfo {year} {2020}{\natexlab{a}})}\BibitemShut {NoStop}%
\bibitem [{\citenamefont {Wolfe}\ \emph {et~al.}(2020)\citenamefont {Wolfe},
  \citenamefont {Schmid}, \citenamefont {Sainz}, \citenamefont {Kunjwal},\ and\
  \citenamefont {Spekkens}}]{wolfe2020quantifying}%
  \BibitemOpen
  \bibfield  {author} {\bibinfo {author} {\bibfnamefont {E.}~\bibnamefont
  {Wolfe}}, \bibinfo {author} {\bibfnamefont {D.}~\bibnamefont {Schmid}},
  \bibinfo {author} {\bibfnamefont {A.~B.}\ \bibnamefont {Sainz}}, \bibinfo
  {author} {\bibfnamefont {R.}~\bibnamefont {Kunjwal}},\ and\ \bibinfo {author}
  {\bibfnamefont {R.~W.}\ \bibnamefont {Spekkens}},\ }\bibfield  {title}
  {\bibinfo {title} {Quantifying {B}ell: the resource theory of nonclassicality
  of common-cause boxes},\ }\href {https://doi.org/10.22331/q-2020-06-08-280}
  {\bibfield  {journal} {\bibinfo  {journal} {{Quantum}}\ }\textbf {\bibinfo
  {volume} {4}},\ \bibinfo {pages} {280} (\bibinfo {year} {2020})}\BibitemShut
  {NoStop}%
\bibitem [{\citenamefont {Schmid}\ \emph
  {et~al.}(2020{\natexlab{b}})\citenamefont {Schmid}, \citenamefont {Fraser},
  \citenamefont {Kunjwal}, \citenamefont {Sainz}, \citenamefont {Wolfe},\ and\
  \citenamefont {Spekkens}}]{schmid2020standard}%
  \BibitemOpen
  \bibfield  {author} {\bibinfo {author} {\bibfnamefont {D.}~\bibnamefont
  {Schmid}}, \bibinfo {author} {\bibfnamefont {T.~C.}\ \bibnamefont {Fraser}},
  \bibinfo {author} {\bibfnamefont {R.}~\bibnamefont {Kunjwal}}, \bibinfo
  {author} {\bibfnamefont {A.~B.}\ \bibnamefont {Sainz}}, \bibinfo {author}
  {\bibfnamefont {E.}~\bibnamefont {Wolfe}},\ and\ \bibinfo {author}
  {\bibfnamefont {R.~W.}\ \bibnamefont {Spekkens}},\ }\href@noop {} {\bibinfo
  {title} {Why standard entanglement theory is inappropriate for the study of
  {Bell} scenarios}} (\bibinfo {year} {2020}{\natexlab{b}}),\ \Eprint
  {https://arxiv.org/abs/2004.09194} {arXiv:2004.09194} \BibitemShut {NoStop}%
\bibitem [{\citenamefont {Vidal}\ and\ \citenamefont
  {Tarrach}(1999)}]{robustness}%
  \BibitemOpen
  \bibfield  {author} {\bibinfo {author} {\bibfnamefont {G.}~\bibnamefont
  {Vidal}}\ and\ \bibinfo {author} {\bibfnamefont {R.}~\bibnamefont
  {Tarrach}},\ }\bibfield  {title} {\bibinfo {title} {Robustness of
  entanglement},\ }\href {https://doi.org/10.1103/PhysRevA.59.141} {\bibfield
  {journal} {\bibinfo  {journal} {Phys. Rev. A}\ }\textbf {\bibinfo {volume}
  {59}},\ \bibinfo {pages} {141} (\bibinfo {year} {1999})}\BibitemShut
  {NoStop}%
\bibitem [{\citenamefont {Chaves}\ \emph {et~al.}(2015)\citenamefont {Chaves},
  \citenamefont {Majenz},\ and\ \citenamefont {Gross}}]{Chaves2015}%
  \BibitemOpen
  \bibfield  {author} {\bibinfo {author} {\bibfnamefont {R.}~\bibnamefont
  {Chaves}}, \bibinfo {author} {\bibfnamefont {C.}~\bibnamefont {Majenz}},\
  and\ \bibinfo {author} {\bibfnamefont {D.}~\bibnamefont {Gross}},\ }\bibfield
   {title} {\bibinfo {title} {Information{\textendash}theoretic implications of
  quantum causal structures},\ }\href {https://doi.org/10.1038/ncomms6766}
  {\bibfield  {journal} {\bibinfo  {journal} {Nat. Comm.}\ }\textbf {\bibinfo
  {volume} {6}} (\bibinfo {year} {2015})}\BibitemShut {NoStop}%
\bibitem [{\citenamefont {Costa}\ and\ \citenamefont
  {Shrapnel}(2016)}]{Costa2016}%
  \BibitemOpen
  \bibfield  {author} {\bibinfo {author} {\bibfnamefont {F.}~\bibnamefont
  {Costa}}\ and\ \bibinfo {author} {\bibfnamefont {S.}~\bibnamefont
  {Shrapnel}},\ }\bibfield  {title} {\bibinfo {title} {Quantum causal
  modelling},\ }\href {https://doi.org/10.1088/1367-2630/18/6/063032}
  {\bibfield  {journal} {\bibinfo  {journal} {New J. Phys.}\ }\textbf {\bibinfo
  {volume} {18}},\ \bibinfo {pages} {063032} (\bibinfo {year}
  {2016})}\BibitemShut {NoStop}%
\bibitem [{\citenamefont {Wolfe}\ \emph {et~al.}(2019)\citenamefont {Wolfe},
  \citenamefont {Pozas-Kerstjens}, \citenamefont {Grinberg}, \citenamefont
  {Rosset}, \citenamefont {Ac{\ifmmode\acute{\imath}\else\'{\i}\fi}n},\ and\
  \citenamefont {Navascues}}]{qinflation}%
  \BibitemOpen
  \bibfield  {author} {\bibinfo {author} {\bibfnamefont {E.}~\bibnamefont
  {Wolfe}}, \bibinfo {author} {\bibfnamefont {A.}~\bibnamefont
  {Pozas-Kerstjens}}, \bibinfo {author} {\bibfnamefont {M.}~\bibnamefont
  {Grinberg}}, \bibinfo {author} {\bibfnamefont {D.}~\bibnamefont {Rosset}},
  \bibinfo {author} {\bibfnamefont {A.}~\bibnamefont
  {Ac{\ifmmode\acute{\imath}\else\'{\i}\fi}n}},\ and\ \bibinfo {author}
  {\bibfnamefont {M.}~\bibnamefont {Navascues}},\ }\href@noop {} {\bibinfo
  {title} {{Quantum Inflation: A General Approach to Quantum Causal
  Compatibility}}} (\bibinfo {year} {2019}),\ \Eprint
  {https://arxiv.org/abs/1909.10519} {arXiv:1909.10519} \BibitemShut {NoStop}%
\bibitem [{\citenamefont {Barrett}\ \emph {et~al.}(2019)\citenamefont
  {Barrett}, \citenamefont {Lorenz},\ and\ \citenamefont
  {Oreshkov}}]{Barrett2019QCM}%
  \BibitemOpen
  \bibfield  {author} {\bibinfo {author} {\bibfnamefont {J.}~\bibnamefont
  {Barrett}}, \bibinfo {author} {\bibfnamefont {R.}~\bibnamefont {Lorenz}},\
  and\ \bibinfo {author} {\bibfnamefont {O.}~\bibnamefont {Oreshkov}},\
  }\href@noop {} {\bibinfo {title} {{Quantum Causal Models}}} (\bibinfo {year}
  {2019}),\ \Eprint {https://arxiv.org/abs/1906.10726} {arXiv:1906.10726}
  \BibitemShut {NoStop}%
\bibitem [{\citenamefont {Pozas-Kerstjens}\ \emph {et~al.}(2019)\citenamefont
  {Pozas-Kerstjens}, \citenamefont {Rabelo}, \citenamefont {Rudnicki},
  \citenamefont {Chaves}, \citenamefont {Cavalcanti}, \citenamefont
  {Navascués},\ and\ \citenamefont {Acín}}]{Pozas2019}%
  \BibitemOpen
  \bibfield  {author} {\bibinfo {author} {\bibfnamefont {A.}~\bibnamefont
  {Pozas-Kerstjens}}, \bibinfo {author} {\bibfnamefont {R.}~\bibnamefont
  {Rabelo}}, \bibinfo {author} {\bibfnamefont {L.}~\bibnamefont {Rudnicki}},
  \bibinfo {author} {\bibfnamefont {R.}~\bibnamefont {Chaves}}, \bibinfo
  {author} {\bibfnamefont {D.}~\bibnamefont {Cavalcanti}}, \bibinfo {author}
  {\bibfnamefont {M.}~\bibnamefont {Navascués}},\ and\ \bibinfo {author}
  {\bibfnamefont {A.}~\bibnamefont {Acín}},\ }\bibfield  {title} {\bibinfo
  {title} {{Bounding the Sets of Classical and Quantum Correlations in
  Networks}},\ }\href {https://doi.org/10.1103/PhysRevLett.123.140503}
  {\bibfield  {journal} {\bibinfo  {journal} {Phys. Rev. Lett.}\ }\textbf
  {\bibinfo {volume} {123}} (\bibinfo {year} {2019})}\BibitemShut {NoStop}%
\bibitem [{\citenamefont {Bowles}\ \emph {et~al.}(2020)\citenamefont {Bowles},
  \citenamefont {Baccari},\ and\ \citenamefont {Salavrakos}}]{Bowles2019}%
  \BibitemOpen
  \bibfield  {author} {\bibinfo {author} {\bibfnamefont {J.}~\bibnamefont
  {Bowles}}, \bibinfo {author} {\bibfnamefont {F.}~\bibnamefont {Baccari}},\
  and\ \bibinfo {author} {\bibfnamefont {A.}~\bibnamefont {Salavrakos}},\
  }\bibfield  {title} {\bibinfo {title} {Bounding sets of sequential quantum
  correlations and device-independent randomness certification},\ }\href
  {https://doi.org/10.22331/q-2020-10-19-344} {\bibfield  {journal} {\bibinfo
  {journal} {{Quantum}}\ }\textbf {\bibinfo {volume} {4}},\ \bibinfo {pages}
  {344} (\bibinfo {year} {2020})}\BibitemShut {NoStop}%
\bibitem [{\citenamefont {{Wolfe}}\ \emph {et~al.}(2019)\citenamefont
  {{Wolfe}}, \citenamefont {{Spekkens}},\ and\ \citenamefont
  {{Fritz}}}]{inflation}%
  \BibitemOpen
  \bibfield  {author} {\bibinfo {author} {\bibfnamefont {E.}~\bibnamefont
  {{Wolfe}}}, \bibinfo {author} {\bibfnamefont {R.~W.}\ \bibnamefont
  {{Spekkens}}},\ and\ \bibinfo {author} {\bibfnamefont {T.}~\bibnamefont
  {{Fritz}}},\ }\bibfield  {title} {\bibinfo {title} {{The Inflation Technique
  for Causal Inference with Latent Variables}},\ }\href
  {https://doi.org/10.1515/jci-2017-0020} {\bibfield  {journal} {\bibinfo
  {journal} {J. Causal Inference}\ }\textbf {\bibinfo {volume} {7}} (\bibinfo
  {year} {2019})}\BibitemShut {NoStop}%
\bibitem [{\citenamefont {Greenberger}\ \emph {et~al.}(1989)\citenamefont
  {Greenberger}, \citenamefont {Horne},\ and\ \citenamefont {Zeilinger}}]{GHZ}%
  \BibitemOpen
  \bibfield  {author} {\bibinfo {author} {\bibfnamefont {D.~M.}\ \bibnamefont
  {Greenberger}}, \bibinfo {author} {\bibfnamefont {M.~A.}\ \bibnamefont
  {Horne}},\ and\ \bibinfo {author} {\bibfnamefont {A.}~\bibnamefont
  {Zeilinger}},\ }\bibinfo {title} {{Going beyond Bell's theorem}},\ in\ \href
  {https://doi.org/10.1007/978-94-017-0849-4_10} {\emph {\bibinfo {booktitle}
  {{Bell's Theorem, Quantum Theory and Conceptions of the Universe}}}},\
  \bibinfo {editor} {edited by\ \bibinfo {editor} {\bibfnamefont
  {M.}~\bibnamefont {Kafatos}}}\ (\bibinfo  {publisher} {Springer
  Netherlands},\ \bibinfo {address} {Dordrecht},\ \bibinfo {year} {1989})\ pp.\
  \bibinfo {pages} {69--72}\BibitemShut {NoStop}%
\bibitem [{\citenamefont {Peres}(1996)}]{PPT}%
  \BibitemOpen
  \bibfield  {author} {\bibinfo {author} {\bibfnamefont {A.}~\bibnamefont
  {Peres}},\ }\bibfield  {title} {\bibinfo {title} {{Separability Criterion for
  Density Matrices}},\ }\href {https://doi.org/10.1103/PhysRevLett.77.1413}
  {\bibfield  {journal} {\bibinfo  {journal} {Phys. Rev. Lett.}\ }\textbf
  {\bibinfo {volume} {77}},\ \bibinfo {pages} {1413} (\bibinfo {year}
  {1996})}\BibitemShut {NoStop}%
\bibitem [{\citenamefont {Grant}\ and\ \citenamefont {Boyd}(2020)}]{cvx}%
  \BibitemOpen
  \bibfield  {author} {\bibinfo {author} {\bibfnamefont {M.}~\bibnamefont
  {Grant}}\ and\ \bibinfo {author} {\bibfnamefont {S.}~\bibnamefont {Boyd}},\
  }\href {http://cvxr.com/cvx} {\bibinfo {title} {{CVX: Matlab Software for
  Disciplined Convex Programming, version 2.2}}},\ \bibinfo {howpublished}
  {\url{http://cvxr.com/cvx}} (\bibinfo {year} {2020})\BibitemShut {NoStop}%
\bibitem [{\citenamefont {{MOSEK~ApS}}(2019)}]{mosek}%
  \BibitemOpen
  \bibfield  {author} {\bibinfo {author} {\bibnamefont {{MOSEK~ApS}}},\ }\href
  {http://docs.mosek.com/9.1/toolbox/index.html} {\bibinfo {title} {{The MOSEK
  optimization toolbox for MATLAB manual}}},\ \bibinfo {howpublished}
  {\url{https://docs.mosek.com}} (\bibinfo {year} {2019})\BibitemShut {NoStop}%
\bibitem [{\citenamefont {D\"ur}\ \emph {et~al.}(2000)\citenamefont {D\"ur},
  \citenamefont {Vidal},\ and\ \citenamefont {Cirac}}]{W_state}%
  \BibitemOpen
  \bibfield  {author} {\bibinfo {author} {\bibfnamefont {W.}~\bibnamefont
  {D\"ur}}, \bibinfo {author} {\bibfnamefont {G.}~\bibnamefont {Vidal}},\ and\
  \bibinfo {author} {\bibfnamefont {J.~I.}\ \bibnamefont {Cirac}},\ }\bibfield
  {title} {\bibinfo {title} {Three qubits can be entangled in two inequivalent
  ways},\ }\href {https://doi.org/10.1103/PhysRevA.62.062314} {\bibfield
  {journal} {\bibinfo  {journal} {Phys. Rev. A}\ }\textbf {\bibinfo {volume}
  {62}},\ \bibinfo {pages} {062314} (\bibinfo {year} {2000})}\BibitemShut
  {NoStop}%
\bibitem [{\citenamefont {Erd\"osi}\ \emph {et~al.}(2013)\citenamefont
  {Erd\"osi}, \citenamefont {Huber}, \citenamefont {Hiesmayr},\ and\
  \citenamefont {Hasegawa}}]{exp1}%
  \BibitemOpen
  \bibfield  {author} {\bibinfo {author} {\bibfnamefont {D.}~\bibnamefont
  {Erd\"osi}}, \bibinfo {author} {\bibfnamefont {M.}~\bibnamefont {Huber}},
  \bibinfo {author} {\bibfnamefont {B.~C.}\ \bibnamefont {Hiesmayr}},\ and\
  \bibinfo {author} {\bibfnamefont {Y.}~\bibnamefont {Hasegawa}},\ }\bibfield
  {title} {\bibinfo {title} {Proving the generation of genuine multipartite
  entanglement in a single-neutron interferometer experiment},\ }\href
  {https://doi.org/10.1088/1367-2630/15/2/023033} {\bibfield  {journal}
  {\bibinfo  {journal} {New J. Phys.}\ }\textbf {\bibinfo {volume} {15}},\
  \bibinfo {pages} {023033} (\bibinfo {year} {2013})}\BibitemShut {NoStop}%
\bibitem [{\citenamefont {H\"ubel}\ \emph {et~al.}(2010)\citenamefont
  {H\"ubel}, \citenamefont {Hamel}, \citenamefont {Fedrizzi}, \citenamefont
  {Ramelow}, \citenamefont {Resch},\ and\ \citenamefont {Jennewein}}]{exp2}%
  \BibitemOpen
  \bibfield  {author} {\bibinfo {author} {\bibfnamefont {H.}~\bibnamefont
  {H\"ubel}}, \bibinfo {author} {\bibfnamefont {D.~R.}\ \bibnamefont {Hamel}},
  \bibinfo {author} {\bibfnamefont {A.}~\bibnamefont {Fedrizzi}}, \bibinfo
  {author} {\bibfnamefont {S.}~\bibnamefont {Ramelow}}, \bibinfo {author}
  {\bibfnamefont {K.~J.}\ \bibnamefont {Resch}},\ and\ \bibinfo {author}
  {\bibfnamefont {T.}~\bibnamefont {Jennewein}},\ }\bibfield  {title} {\bibinfo
  {title} {Direct generation of photon triplets using cascaded photon-pair
  sources},\ }\href {https://doi.org/10.1038/nature09175} {\bibfield  {journal}
  {\bibinfo  {journal} {Nature}\ }\textbf {\bibinfo {volume} {466}},\ \bibinfo
  {pages} {601} (\bibinfo {year} {2010})}\BibitemShut {NoStop}%
\bibitem [{\citenamefont {Resch}\ \emph {et~al.}(2005)\citenamefont {Resch},
  \citenamefont {Walther},\ and\ \citenamefont {Zeilinger}}]{exp3}%
  \BibitemOpen
  \bibfield  {author} {\bibinfo {author} {\bibfnamefont {K.~J.}\ \bibnamefont
  {Resch}}, \bibinfo {author} {\bibfnamefont {P.}~\bibnamefont {Walther}},\
  and\ \bibinfo {author} {\bibfnamefont {A.}~\bibnamefont {Zeilinger}},\
  }\bibfield  {title} {\bibinfo {title} {{Full Characterization of a
  Three-Photon {Greenberger-Horne-Zeilinger} State Using Quantum State
  Tomography}},\ }\href {https://doi.org/10.1103/PhysRevLett.94.070402}
  {\bibfield  {journal} {\bibinfo  {journal} {Phys. Rev. Lett.}\ }\textbf
  {\bibinfo {volume} {94}},\ \bibinfo {pages} {070402} (\bibinfo {year}
  {2005})}\BibitemShut {NoStop}%
\bibitem [{\citenamefont {Walther}\ \emph {et~al.}(2005)\citenamefont
  {Walther}, \citenamefont {Resch},\ and\ \citenamefont {Zeilinger}}]{exp4}%
  \BibitemOpen
  \bibfield  {author} {\bibinfo {author} {\bibfnamefont {P.}~\bibnamefont
  {Walther}}, \bibinfo {author} {\bibfnamefont {K.~J.}\ \bibnamefont {Resch}},\
  and\ \bibinfo {author} {\bibfnamefont {A.}~\bibnamefont {Zeilinger}},\
  }\bibfield  {title} {\bibinfo {title} {{Local Conversion of
  {Greenberger-Horne-Zeilinger} States to Approximate $W$ States}},\ }\href
  {https://doi.org/10.1103/PhysRevLett.94.240501} {\bibfield  {journal}
  {\bibinfo  {journal} {Phys. Rev. Lett.}\ }\textbf {\bibinfo {volume} {94}},\
  \bibinfo {pages} {240501} (\bibinfo {year} {2005})}\BibitemShut {NoStop}%
\bibitem [{\citenamefont {Werner}\ and\ \citenamefont {Wolf}(2001)}]{seeSaw1}%
  \BibitemOpen
  \bibfield  {author} {\bibinfo {author} {\bibfnamefont {R.~F.}\ \bibnamefont
  {Werner}}\ and\ \bibinfo {author} {\bibfnamefont {M.~M.}\ \bibnamefont
  {Wolf}},\ }\bibfield  {title} {\bibinfo {title} {All multipartite {B}ell
  correlation inequalities for two dichotomic observables per site},\ }\href
  {https://doi.org/10.1103/physreva.64.032112} {\bibfield  {journal} {\bibinfo
  {journal} {Phys. Rev. A}\ }\textbf {\bibinfo {volume} {64}},\ \bibinfo
  {pages} {1} (\bibinfo {year} {2001})}\BibitemShut {NoStop}%
\bibitem [{\citenamefont {P\'al}\ and\ \citenamefont
  {V\'ertesi}(2010)}]{seeSaw2}%
  \BibitemOpen
  \bibfield  {author} {\bibinfo {author} {\bibfnamefont {K.~F.}\ \bibnamefont
  {P\'al}}\ and\ \bibinfo {author} {\bibfnamefont {T.}~\bibnamefont
  {V\'ertesi}},\ }\bibfield  {title} {\bibinfo {title} {Maximal violation of a
  bipartite three-setting, two-outcome {Bell} inequality using
  infinite-dimensional quantum systems},\ }\href
  {https://doi.org/10.1103/PhysRevA.82.022116} {\bibfield  {journal} {\bibinfo
  {journal} {Phys. Rev. A}\ }\textbf {\bibinfo {volume} {82}},\ \bibinfo
  {pages} {022116} (\bibinfo {year} {2010})}\BibitemShut {NoStop}%
\bibitem [{\citenamefont {Kraft}\ \emph {et~al.}(2020)\citenamefont {Kraft},
  \citenamefont {Designolle}, \citenamefont {Ritz}, \citenamefont {Brunner},
  \citenamefont {Gühne},\ and\ \citenamefont {Huber}}]{kraft2020quantum}%
  \BibitemOpen
  \bibfield  {author} {\bibinfo {author} {\bibfnamefont {T.}~\bibnamefont
  {Kraft}}, \bibinfo {author} {\bibfnamefont {S.}~\bibnamefont {Designolle}},
  \bibinfo {author} {\bibfnamefont {C.}~\bibnamefont {Ritz}}, \bibinfo {author}
  {\bibfnamefont {N.}~\bibnamefont {Brunner}}, \bibinfo {author} {\bibfnamefont
  {O.}~\bibnamefont {Gühne}},\ and\ \bibinfo {author} {\bibfnamefont
  {M.}~\bibnamefont {Huber}},\ }\href@noop {} {\bibinfo {title} {Quantum
  entanglement in the triangle network}} (\bibinfo {year} {2020}),\ \Eprint
  {https://arxiv.org/abs/2002.03970} {arXiv:2002.03970} \BibitemShut {NoStop}%
\bibitem [{\citenamefont {Luo}(2020)}]{luo}%
  \BibitemOpen
  \bibfield  {author} {\bibinfo {author} {\bibfnamefont {M.-X.}\ \bibnamefont
  {Luo}},\ }\href@noop {} {\bibinfo {title} {New genuine multipartite
  entanglement}} (\bibinfo {year} {2020}),\ \Eprint
  {https://arxiv.org/abs/2003.07153} {arXiv:2003.07153} \BibitemShut {NoStop}%
\end{thebibliography}%

\appendix
\section{Appendix A: analytic witnesses for \\genuine network tripartite entanglement}\label{sec:simpleproofs}
The goal of this appendix is to prove the following result:
\begin{theo}
\label{bounds}
Let $\rho_{\A\B\C}\in \mathsf{B}({\mathbb C}^d\otimes {\mathbb C}^d\otimes {\mathbb C}^d)$ be a tripartite quantum state, and let $a,b,c$ be the outcomes which result when we locally probe subsystems $\A,\B,\C$. If $\rho_{\A\B\C}$ is not genuinely network $3$-entangled, then it must satisfy the relations:
\begin{align}
&\bra{\mathsf{GHZ}_d}\rho\ket{\mathsf{GHZ}_d}\leq \frac{2d (3 d + \sqrt{2d-1})}{1-2 d + 9 d^2},\label{result1}\\
&H(a:b) + H(b:c)-H(b)\leq S(\A)+S(\A|\B\C),
\label{result2}
\end{align}
\noindent where $S(\A)$, $S(\A|\B\C)$ respectively denote the von Neumann entropy of $\rho_{(\A)}$ and the conditional entropy of system $\A$ with respect to $\B\C$, i.e., $S(\rho_{\A\B\C})-S(\rho_{(\B\C)})$.
\end{theo}

\noindent The intuition behind the proofs of both inequalities is the same. First, we assume that the state $\rho_{\A\B\C}$ admits the six-partite extensions $\tau_{\A_1 \B_1 \C_1 \A_2 \B_2 \C_2}$ and $\gamma_{\A_3 \B_3 \C_3 \A_4 \B_4 \C_4}$ described in the main text. Then we prove that a strong correlation between the random variables $a_3, b_3$ and $b_3, c_3$ implies a strong correlation between the variables $a_3, c_3$, and hence a strong correlation between the variables $a_1, c_2$. Next, we show that the correlation between the variables $a_1, c_2$ is upper bounded in some way by the purity of the original state $\rho_{\A\B\C}$. To obtain one bound or another we rely on different measures of correlation and purity.

The following lemma will establish the transitivity of strongly correlated variables.
\begin{lemma}
\label{transitivity}
Let $x,y,z$ be jointly distributed random variables. Then, the following inequalities hold:
\begin{align}
&P(x=z)\geq P(x=y)+P(y=z)-1,\label{trans1}\\
&H(x:z)\geq H(x:y)+H(y:z)-H(y).\label{trans2}
\end{align}
\end{lemma}
\begin{proof}
Let $P(x,y,z)$ be the joint probability distribution of the three variables. Then we have that
\begin{align}
&P(x=y)+P(y=z)-P(x=z)=\sum_{i}P(i,i,i)+\nonumber\\
&\sum_{j\not=i}P(i,i,j)+P(j,i,i)-P(i,j,i)\nonumber\\
&\leq \sum_{k}P(i,i,i)+\sum_{j\not=i}P(i,i,j)+P(j,i,i).
\end{align}

\noindent The right-hand side of the above equation contains the probabilities of a set of incompatible events. Its sum is thus bounded by $1$, hence proving inequality (\ref{trans1}).

To prove Eq.~\eqref{trans2}, we invoke strong subadditivity. Namely, for any three random variables $x,y,z$, it holds that
\be
H(x,y,z)\leq H(x,y)+H(y,z)-H(y).
\ee
\noindent The left-hand side of the equation above can be lower bounded by $H(x,z)$. It follows that $H(x:z)=H(x)+H(z)-H(x,z)$ is lower bounded by $H(x)+H(z)-H(x,y)-H(y,z)+H(y)$. This, in turn, equals the right-hand side of Eq.~\eqref{trans2}.
\end{proof}

The next lemma will relate the purity of a tripartite state with the correlations it can establish with other systems.

\begin{lemma}
\label{weak_corre}
Consider a four-partite quantum state $\sigma_{\A\B\C\Y}$, with $F=\bra{\mathsf{GHZ}_d}\sigma_{(\A\B\C)}\ket{\mathsf{GHZ}_d}$, and suppose that $a,y$ are the result of measuring systems $\A,\Y$ in the computational basis, then the inequality
\be
P(a=y)\leq 1+\left(\frac{1}{d}{-}1\right)F +2\sqrt{\frac{F(1{-}F)}{d}}
\label{weak1}
\ee
\noindent holds. Moreover, independently of the nature of the measurements, the relation
\be
H(a:y)\leq S(\A)+S(\A|\B\C)
\label{weak2}
\ee
\noindent is satisfied.
\end{lemma}

\begin{proof}
Suppose that we measure systems $\A$ and $\Y$ in the computational basis, and define the operator
\be
E\equiv \sum_{i=1}^d\proj{i}\otimes \id^{\otimes 2}\otimes \proj{i}.
\ee
\noindent Then, $P(a=y)=\tr[E\sigma]$. Furthermore, one can verify that
\be
(P_0\otimes \id)E(P_0\otimes \id)=\frac{1}{d}\proj{\mathsf{GHZ}_d}\otimes\id,
\label{S_princ}
\ee
\noindent where $P_0,P_1$ are the projectors defined by $P_0=\proj{\mathsf{GHZ}_d}$, $P_1=\id-P_0$.

We have that
\be
\tr\left[E\sigma\right]=\sum_{i,j=0,1}\omega_{ij},
\ee
\noindent where $\omega$ is the $2\times 2$ matrix defined by
\be
\omega_{ij}=\tr\big[\sigma (P_i\otimes \id) E (P_j\otimes \id)\big].
\ee
\noindent $\omega$ is positive semidefinite. Indeed, take an arbitrary vector $\ket{c}$. Then,
\be
\bra{c}\omega\ket{c}=\tr\left[\sigma \left(\sum_ic^*_iP_i\otimes \id\right)E\left(\sum_jc_jP_j\otimes \id\right)\right]\geq 0,
\ee
\noindent where the last inequality stems from the fact that both $\sigma $and $E$ are positive semidefinite.

From the positive-semidefiniteness of $\omega$ it follows that $|\omega_{01}|\leq \sqrt{\omega_{00}\omega_{11}}$. On the other hand, by (\ref{S_princ}) we have that
\be
\omega_{00}=\frac{1}{d}\tr\left[\sigma\left(\proj{\mathsf{GHZ}_d}\otimes\id_d\right)\right]=\frac{F}{d}.
\ee
\noindent In addition, ${\omega_{11}=\tr[\tilde{\sigma}E]}$, where $\tilde{\sigma}$ is the positive semidefinite operator defined by
\be
\tilde{\sigma}\equiv (P_1\otimes\id)\sigma (P_1\otimes\id).
\ee
\noindent Note that ${\tr\left[\tilde{\sigma}\right]=\tr\left[\sigma_{(\A\B\C)}P_1\right]=1{-}F}$. Since the operator $E$ has norm $1$, it follows that ${\omega_{11}=\tr\left[\tilde{\sigma}E\right]\leq 1{-}F}$. Putting all together, we have that
\be
P(a=y)\leq \frac{F}{d}+ 1-F+2\sqrt{\frac{F(1-F)}{d}}.
\ee
\noindent This proves Eq. \eqref{weak1}.

Let ${\alpha,\text{ }\beta}$ be two quantum systems. By ${S(\alpha|\beta)=S(\alpha \beta)-S(\beta)}$ we denote the conditional quantum information; by ${S(\alpha:\beta)=S(\alpha)+S(\beta)-S(\alpha \beta)}$, the quantum mutual information. To prove Eq. \eqref{weak2}, we invoke weak monotonicity, namely, the fact that for any three quantum subsystems $\alpha,\beta,\gamma$, ${S(\alpha|\beta)+S(\alpha|\gamma)\geq 0}$.  Taking ${\alpha=\A}$, ${\beta=\B\C}$, ${\gamma =\Y}$, we have that ${-S(\A|\Y)\leq S(\A|\B\C)}$. By the data processing inequality it thus follows that
\begin{align}
H(a:y)&\leq S(\A:\Y)=S(\A)-S(\A|\Y)\\\nonumber
 &\leq S(\A)+S(\A|\B\C).\tag*{\qedhere}
\end{align}
\end{proof}

Having reached this point, we are ready to prove part of Theorem \ref{bounds}. Choose Positive Operator Valued Measures (POVMs) $M_\A$, $M_\B$, $M_\C$ and use them to probe the type-$\A$, type-$\B$ and type-$\C$ subsystems of $\gamma$ and $\tau$, thus obtaining the random variables $a_1,a_2,a_3,a_4$, $b_1,b_2,b_3,b_4$, $c_1,c_2,c_3,c_4$. From the constraints $\rho_{(\A\B)}=\gamma_{(\A_3\B_3)}$, $\rho_{(\B\C)}=\gamma_{(\B_3\C_3)}$ and Lemma \ref{transitivity}, we arrive at the relations
\begin{align}
&P(a=b) + P(b=c) - 1=\nonumber\\
&P(a_3=b_3) + P(b_3=c_3) - 1\leq P(a_3=c_3),
\label{interm1}
\end{align}
\noindent and
\begin{align}
&H(a:b)+H(b:c)-H(b)=\nonumber\\
&H(a_3:b_3)+H(b_3:c_3)-H(b_3)\leq H(a_3:c_3),
\label{interm2}
\end{align}
\noindent where $a,b,c$ is the result of locally measuring $\rho_{\A\B\C}$ according to the POVMs $M_\A$, $M_\B$, and $M_\C$.

On the other hand, ${\gamma_{(\A_3\C_3)}=\tau_{(\A_1\C_2)}}$, and hence ${H(a_3:c_3)=H(a_1:c_2)}$ and ${P(a_3=c_3)=P(a_1=c_2)}$. Invoking Lemma \ref{weak_corre}, the right hand side of Eq.~\eqref{interm2} is upper bounded by the right-hand side of Eq. \eqref{weak2}. This proves Eq. \eqref{result2}.

If $M_\A$, $M_\B$, $M_\C$ moreover correspond to measurements in the computational basis, then we can invoke again Lemma~\ref{weak_corre} to bound the right-hand side of Eq. \eqref{interm1} with the right-hand side of Eq. \eqref{weak1}. This gives:
\begin{align}\begin{split}
P(a=b) + P(b=c&) - 1\leq\\& 1+\left(\frac{1}{d}-1\right)F +2\sqrt{\frac{F(1-F)}{d}}.
\label{casi}
\end{split}\end{align}
\noindent To prove Eq. \eqref{result1}, we need to lower bound $p(a=b)$ and $p(b=c)$ in terms of the \textsf{GHZ} fidelity. The necessary bound is provided by the next lemma.

\begin{lemma}
\label{strong_corre}
Let $\rho_{\A\B\C}\in \mathsf{B}({\mathbb C}^d\otimes {\mathbb C}^d\otimes {\mathbb C}^d)$ be a tripartite quantum state, and let $F=\bra{\mathsf{GHZ}_d}\rho\ket{\mathsf{GHZ}_d}$ be its \textup{\textsf{GHZ}} fidelity. If we measure any two systems $\A,\B,\C$ in the computational basis, the corresponding random variables $a,b$ will satisfy:
\be
P(a=b)\geq F.
\ee
\end{lemma}
\begin{proof}
Without loss of generality, suppose that we measure systems $\A$ and $\B$, obtaining the random variables $a$ and $b$, respectively. Notice that $P(a=b)=\tr\left[S\rho\right]$, with $S$ defined by
\be
S\equiv \sum_{i=1}^d \proj{i}^{\otimes 2}\otimes \id_d.
\ee
\noindent Then one can verify that
\begin{align}\begin{split}
&P_0SP_0=P_0=\proj{\mathsf{GHZ}_d},\\& P_1SP_0=P_0SP_1 =0.
\end{split}\end{align}
\noindent It follows that
\begin{align}
P(a=b)&=\tr\left[S\rho\right]=\sum_{i,j=0,1}\tr\left[\rho P_i S P_j\right]=\\\nonumber
&\tr\big[\proj{\mathsf{GHZ}_d}\rho\big]+\tr\left[P_1SP_1\rho\right]\geq F.\tag*{\qedhere}
\end{align}
\end{proof}

Now, use the previous lemma to lower bound the left-hand side of Eq. \eqref{casi} by $2F-1$. Solving the inequality for $F$, we arrive at Eq. \eqref{result1}.

\section{Appendix B: analytic witnesses for \\genuine network $k$-entanglement}\label{sec:generalproofs}

The arguments leading to Theorem \ref{bounds} can be extended to detect network $k$-entanglement. Consider this time a $k$-partite quantum state $\rho_{\X^0\dots\X^{k{\shortminus}1}}$, and suppose that it can be generated by applying correlated local maps to $(k{\shortminus}1)$-partite quantum states of the form ${\{\sigma^{0\oplus j,\dots,(k{\shortminus}2) \oplus j}:j=0,\dots,k{\shortminus}1\}}$, where $\oplus$ indicates addition modulo $k$. As in the tripartite case, we consider a double network with nodes ${\X^0_1, \X^1_1, \dots, \X^{k{\shortminus}1}_1, \X^0_2, \X^1_2, \dots, \X^{k{\shortminus}1}_2}$. For ${j=0,\dots,k{\shortminus}1}$, we distribute two copies of the states ${\sigma^{0\oplus j,\dots,(k{\shortminus}2) \oplus j}}$ to systems ${\X^j_1, \X^{1\oplus j}_1, \dots, \X^{(k{\shortminus}1)\oplus j}_1}$, ${\X^j_2, \X^{1\oplus j}_2, \dots, \X^{(k{\shortminus}1)\oplus j}_2}$, respectively. By applying the maps $\Omega_{\X^i}^{\lambda}$ over systems $\X_1^i$, $\X_2^i$ and averaging over $\lambda$, we obtain a $2k$-partite quantum state $\tau_{\X^0_1  \dots \X^{k{\shortminus}1}_1 \X^0_2 \dots \X^{k{\shortminus}1}_2}$ with $\tau_{(\X^0_1 \dots \X^{k{\shortminus}1}_1)}=\tau_{(\X^0_2 \dots \X^{k{\shortminus}1}_2)}=\rho_{\X^0\dots\X^{k{\shortminus}1}}$.

On the contrary, consider the systems ${\X^0_3, \X^1_3, \dots, \X^{k{\shortminus}1}_3, \X^0_4, \X^2_4, \dots, \X^{k{\shortminus}1}_4}$, order them as ${0,1,\dots,2k{\shortminus}1}$, and distribute each state ${\sigma^{0\oplus j,\dots,(k{\shortminus}2) \oplus j}}$ to the systems ${{0{+}j (\mbox{mod } 2k)},\dots,k{-}1{+}j (\mbox{mod }2k)}$, for ${j=0,\dots,2k{\shortminus}1}$. Applying the maps $\Omega_{\X^i}^{\lambda}$ and averaging over $\lambda$, we end up with a $2k$-partite state $\gamma_{0,\dots,2k{\shortminus}1}$ with the following properties:

\begin{compactenum}
\item $\gamma_{(i,i+1)}=\rho_{(\X^i\X^{i\oplus 1})}$. This is so because, in the previous construction, each node $i$ shares with node ${i{+}1}$ the ${k{-}2}$ states ${\{\sigma^{\dots,i,i{+}1},\dots, \sigma^{i, i{+}1,\dots}\}}$. These are all the states which those two nodes would have shared had they been part of the network that built $\rho_{\X^0\dots\X^{k{\shortminus}1}}$. Hence, their joint state must correspond to to the latter's reduced state $\rho_{(\X^i\X^{i\oplus 1})}$.

\item $\gamma_{\X_3^{0}\X_3^{k{\shortminus}1}}=\tau_{\X_1^0\X_2^{k{\shortminus}1}}$. This follows from the fact that the states used to generate $\gamma$ were distributed in such a way that no states are shared by systems $i$ and ${i{+}k{-}1 (\mbox{mod }2k)}$.
\end{compactenum}

We are ready to derive new witnesses. Say that the original state $\rho_{\X^0\dots\X^{k{\shortminus}1}}$ has a high fidelity with the $k$-partite \textsf{GHZ} state, i.e.,
\begin{align}
F_{\mathsf{GHZ}_d^k}&\equiv \bra{\mathsf{GHZ}_d^k}\rho\ket{\mathsf{GHZ}_d^k},\\\nonumber\;\text{where}\quad\ket{\mathsf{GHZ}_d^k}&=\sum_{j=1}^d \frac{\ket{j}^{\otimes k}}{\sqrt{d}}.
\end{align}

We locally measure $\rho$ in some basis, obtaining the random variables ${x^0,\dots,x^{k{\shortminus}1}}$. If $F_{\mathsf{GHZ}_d^k}$ is high enough; and the measurement basis is close to the computational one, then one should expect to find a high correlation between $x^i$, $x^j$, for ${i,j =0,\dots,k{\shortminus}1}$. This implies that a measurement of $\gamma$ in the computational basis will produce random variables ${x^0_3,\dots,x^{k{\shortminus}1}_4}$ with very high coincidence probability $P(x^i_3=x^{i+1}_3)$ and mutual information $H(x_3^i:x_3^{i+1})$ between neighboring sites. Applying Lemma \ref{transitivity} recursively, and taking into account that the distributions of $x_3^i$, $x_3^{i+1}$ and $x^i$, $x^{i+1}$ are the same, we have that
\begin{align}
&\sum_{i=0}^{k{\shortminus}2}H(x^i:x^{i+1})-\sum_{i=1}^{k{\shortminus}2} H(x^i)\leq H(x_3^0:x_3^{k{\shortminus}1})\\
\text{and}\quad&\sum_{i=0}^{k{\shortminus}2}P(x^i=x^{i+1})-k+2\leq P(x_3^0=x^{k{\shortminus}1}).
\end{align}

\noindent In turn, the right-hand sides of the above equations respectively equal $H(x_1^0:x_2^{k{\shortminus}1})$ and $P(x_1^0=x_2^{k{\shortminus}1})$.

The proofs of Lemmas \ref{weak_corre} and \ref{strong_corre} easily generalize to the case of $k$ parties: it amounts to replacing expressions such as $\proj{i}\otimes \id^{\otimes 2}\otimes\proj{i}$ by $\proj{i}\otimes \id^{\otimes k{\shortminus}1}\otimes\proj{i}$, and redefining systems $\A,\B,\C,\Y$ as $\A = \X_1^0$, $\B=\X_1^1$, $\C=\X_1^2\dots\X_1^{k{\shortminus}1}$, $\Y=\X_2^{k{\shortminus}1}$. Putting all together, we arrive at the entropic inequality
\begin{align}\begin{split}\label{eq:entropicgeneral}
\sum_{i=0}^{k{\shortminus}2}H(x^i:x^{i+1})-&\sum_{i=1}^{k{\shortminus}2} H(x^i)\leq\\ &S(\X^0)+S(\X^0|\X^1\dots\X^{k{\shortminus}1}),
\end{split}
\end{align}
\noindent valid for arbitrary physical measurements of subsystems ${\X^0,\dots,\X^{k{\shortminus}1}}$, and at
\begin{align}
&(k{-}1)F-k+2\leq 1+\left(\frac{1}{d}{-}1\right)F+2\sqrt{\frac{F(1{-}F)}{d}},
\end{align}
where $F$ is shorthand for $F_{\mathsf{GHZ}_d^k}$.

Equivalently, one can solve for $F_{\mathsf{GHZ}_d^k}$ and write the last witness in linear form as
\begin{align}\label{eq:fidelitygeneral}
 F_{\mathsf{GHZ}_d^k}\leq \displaystyle\frac{d \left(3 - k (d {+} 1) +k^2 d + 2\sqrt{2 + k (d {-} 1) - d}\right)}{1 + 4 d - 2 dk + k^2 d^2},
\end{align}
where inequalities \eqref{eq:entropicgeneral} and \eqref{eq:fidelitygeneral} are satisfied by all states not genuinely network $k$-entangled.

\end{document}